\documentclass[12pt]{article}
\usepackage{amsmath}
\usepackage{amssymb}
\usepackage{graphicx,psfrag,epsf}
\usepackage{enumerate}

\usepackage{tabularray}

\usepackage{booktabs}
\UseTblrLibrary{booktabs}

\usepackage{url} 
\usepackage{float}
\usepackage[table,xcdraw,dvipsnames]{xcolor}

\DeclareMathOperator{\supp}{supp}

\DeclareMathOperator*{\argmin}{arg\,min}

\usepackage{cite}

\usepackage{multirow}
\usepackage{lipsum}
\usepackage{amsfonts}

\usepackage{graphicx}   

\usepackage{amsthm}
\RequirePackage[colorlinks,citecolor=blue,urlcolor=blue]{hyperref}
\usepackage{amssymb,calc}
\usepackage{thmtools}
\usepackage{mathrsfs}
\usepackage{mathtools}
\usepackage{bm}
\usepackage{bbm}
\usepackage{caption}
\usepackage{subcaption}

\usepackage{enumerate}
\usepackage{rotating}
\RequirePackage{cleveref}
\usepackage{hyphenat}
\usepackage[numbers]{natbib}
\usepackage{caption,subcaption}
\usepackage{tikz}
\usepackage{pgfplots}

\usepackage[T1]{fontenc}
\usepackage{textcomp} 
\usepackage{lmodern}

\usepackage{algorithm}
\usepackage[noend]{algpseudocode}
\makeatletter
\renewcommand{\ALG@name}{Algorithm}
\makeatother

\usepackage{titletoc}

\pgfplotsset{width=10cm}

\tikzset{declare function={gamma(\x)=sqrt(2*pi)*\x^(\x-0.5)*exp(-\x)*exp(1/(12*\x));}}
\tikzset{declare function={tpdf(\x,\nu)=gamma(0.5*(\nu+1))/(sqrt(pi*\nu)*gamma(\nu/2))*(1+\x^2/\nu)^(-(\nu+1)/2);}}
\tikzset{declare function={invgampdf(\x,\a,\b)=(\b/\x)^\a/\x/gamma(\a)*exp(-\b/\x);}}

\newcommand{\nhphantom}[1]{\ifmmode\settowidth{\dimen0}{$#1$}\else\settowidth{\dimen0}{#1}\fi\hspace*{-\dimen0}}

\usetikzlibrary{patterns}
\tikzset{
	hatch distance/.store in=\hatchdistance,
	hatch distance=5pt,
	hatch thickness/.store in=\hatchthickness,
	hatch thickness=0.5pt,
}

\newcommand{\wh}[1]{\widehat{#1}}

\definecolor{pink}{rgb}{0.9, 0.17, 0.31}

\def\C {\,|\:}

\newcommand\E{\mathbb E}

\renewcommand\d{\mathrm d}

\renewcommand\P{\mathbb P}

\newcommand\R{\mathbb R}
\renewcommand\b{\bm{\beta}}

\newcommand\iid{\overset{\mathrm{iid}}{\sim}}

\newtheorem{assumption}{Assumption}
\newtheorem{lemma}{Lemma}

\renewcommand{\nhphantom}[1]{\ifmmode\settowidth{\dimen0}{$#1$}\else\settowidth{\dimen0}{#1}\fi\hspace*{-\dimen0}}

\numberwithin{equation}{section}

\newtheorem{thm}{Theorem}[section]

\newtheorem*{defn}{Definition}

\crefname{thm}{Theorem}{Theorems}
\crefname{prop}{Proposition}{Propositions}
\crefname{lem}{Lemma}{Lemmas}
\crefname{coro}{Corollary}{Corollaries}
\crefname{add}{Addendum}{Addendums}
\crefname{asm}{Assumption}{Assumptions}
\crefname{alg}{Algorithm}{Algorithms}
\crefname{proc}{Procedure}{Procedures}
\crefname{exe}{Exercise}{Exercises}
\crefname{exa}{Example}{Examples}
\crefname{prob}{Problem}{Problems}
\crefname{section}{Section}{Sections}
\crefname{subsection}{Section}{Sections}
\crefname{appendix}{Appendix}{Appendices}

\allowdisplaybreaks[3]

\begin{document}

\def\spacingset#1{\renewcommand{\baselinestretch}%
{#1}\small\normalsize} \spacingset{1.1}

\title{\sf Generative Regression with IQ-BART}

\author{Sean O'Hagan and Veronika Ro\v{c}kov\'{a}}


	\maketitle

\bigskip
\begin{abstract}

  Implicit Quantile BART (IQ-BART) posits a non-parametric Bayesian model on the conditional quantile function, acting as a model over a conditional model for $Y$ given $X$. 
 One of the key ingredients is  augmenting the observed data $\{(Y_i,\bm X_i)\}_{i=1}^n$ with uniformly sampled  values $\tau_i$ for $1\leq i\leq n$ which serve as training data for quantile function estimation. Using the fact that the location parameter $\mu$ in a  $\tau$-tilted asymmetric Laplace distribution corresponds to the $\tau^{th}$ quantile, we build a check-loss  likelihood  targeting $\mu$ as the parameter of interest. We equip the  check-loss   likelihood parametrized by $\mu=f(\bm X,\tau)$ with a BART prior on $f(\cdot)$, allowing the conditional quantile function to vary both in $\bm X$ and $\tau$.  The posterior distribution over $\mu(\tau,\bm X)$ can be then distilled for estimation of the {\em entire quantile function} as well as for assessing uncertainty through the variation of posterior draws. 
 Simulation-based predictive inference is immediately available through inverse transform sampling using the learned quantile function.
  The sum-of-trees structure over the conditional quantile function enables flexible distribution-free regression with theoretical guarantees. As a byproduct, we investigate posterior mean quantile estimator as an alternative to the routine sample (posterior mode) quantile estimator. We demonstrate the power of IQ-BART on time series forecasting datasets where IQ-BART can capture multimodality in predictive distributions that might be otherwise missed using traditional parametric approaches.

\end{abstract}

\noindent%

\spacingset{1.45} 

\clearpage

\section{Introduction}

Many contemporary applications involve models whose likelihood functions  are defined only implicitly  through generative  programs. Such  systems are parametrized probabilistic mechanisms that take noise variables as input and   return simulated data  as output \cite{gutmann}.
In this work, we will assume a particular form of a    {\em generative regression} model
\begin{equation}\label{npr}
Y_i=f_{\theta_0}(\bm x_i,\varepsilon_i),\quad \varepsilon_i\iid P_0,
\end{equation}
where each output variable $Y_i\in\R$ for $1\leq i\leq n$ is related to a set of $p$ observable  covariates $\bm x_i=(x_{i1},\dots,x_{ip})'$ and  an unobservable noise variable $\varepsilon_i$   through a push-forward mapping $f_{\theta_0}(\cdot)$.
The  mapping $f_{\theta_0}(\cdot)$ may depend on an unknown parameter of interest $\theta_0\in\Theta\subseteq \R^d$ or we can equivalently treat $f$ itself as a (potentially infinite-dimensional) parameter in which case we write $f(\cdot)$. The  implicit generative model \eqref{npr} presents a substantial generalization of traditional statistical models (such as linear regression) that additively separate signal from noise through a mapping $f_{\theta_0}(\cdot)$ that is linear.
The implicit model \eqref{npr} captures the conditional distribution of $Y_i$ given $\bm X_i=\bm x_i\in\mathcal X$ (for $1\leq i\leq n$) through  $f_{\theta_0}(\cdot)$.
This  mapping is pre-specified in contexts when there is no uncertainty about the data  generating algorithm beyond the value of $\theta_0$.
When there is such added uncertainty and/or there are concerns about model mis-specification,
the  map $f_{\theta}(\cdot)$ itself becomes the target for estimation and inference. We refer This is the focus of our paper.

When $\varepsilon_i\iid P_0$ in \eqref{npr} are uniformly distributed according to $U[0,1]$, the mapping $f_{\theta_0}(\cdot)$ corresponds to the quantile function of the conditional distribution of $Y_i$ given $\bm X_i$. Inference on $f_{\theta_0}(\cdot)$ then boils down to quantile learning which presents an attractive alternative to non-parametric regression approaches based on the conditional density function \citep{orlandi2021densityregression, li2023adaptive}.  Related generative approaches to conditional density regression through quantile function estimation with neural networks were considered by \citet{dabney2018implicit} and \citet{shen2024distributional}.
The second approach generalizes naturally to multivariate responses by assessing  distributional fit through an energy score \citep{gneiting2007strictly},   a popular scoring rule for evaluating multivariate distributional predictions.
\citet{dabney2018implicit}, on the other hand,  introduced a data augmentation scheme for learning the conditional quantile function using neural networks in the context of a reinforcement learning problem. Each sample of training data $\bm X_{i} $ is paired with a uniformly sampled quantile $\tau_{i}$, and the estimation loss for that datapoint is the quantile loss (pinball loss) for the corresponding $\tau_{i}$. This approach can be thought of as randomly assigning data points to be used for the estimation of various quantiles, which allows stochastic gradient descent to learn a smooth function that estimates conditional quantiles. This approach has been exploited in the context of posterior simulation by \citet{polson23_gener_ai_bayes_comput}.

We investigate a fully Bayesian  approach to learning the {\em entire conditional quantile function} non-parametrically assuming that $f_{\theta_0}(\cdot)$ can be closely approximated by some $f_{\theta}\in\mathcal F=\{f_{\theta}(\cdot):\theta\in\Theta\}$, where  $\mathcal F$ is a flexible function class with good approximability properties. This paper will explore functions  $f_\theta\in\mathcal F$ that take a form of a forest (i.e. an aggregation of trees).
We will treat forests in a fully Bayesian way through priors on $\mathcal F$ in the spirit of Bayesian additive regression trees (BART) \citep{chipman2010bart}.
Trees and forests have been deployed for conditional density regression before. \citet{meinshausen2006quantile} observed that the weighted average of empirical quantiles inside leaf nodes among regression trees in a random forest can be performant estimators of the corresponding conditional quantile function. This allows for arbitrary quantiles $q\in(0,1)$ to be estimated from the data.
 In \citet{orlandi2021densityregression}, BART priors are used in functions for conditional density regression through a covariate dependent mixture model.
\citet{li2023adaptive}  posit another model for conditional density estimation using BART by tilting a
baseline Gaussian linear model by a functional  consisting of  a sum of soft decision trees, weighted by random basis functions.

Our work leverages the well-known observation \citep{tsionas2003bayesianquantileinference} that  a  quantile $\tau\in(0,1)$ can be estimated by the maximum likelihood estimator  of a location parameter $\mu$ under the asymmetric Laplace distribution (tilted by $\tau$). This allows one to construct a probabilistic wrapper around quantile learning to perform a fully Bayesian inference using an auxiliary Laplace likelihood.
Indeed, the likelihood is only auxiliary because  it  is {\em not meant  to model  the data}, but rather to identify the inferential target.
The  asymmetric Laplace likelihood, parametrized by the location parameter $\mu=f_\theta(\bm X,\tau)$, will be equipped with a  prior on a set of forests $\mathcal F=\{f_\theta(\cdot):\theta\in\Theta\}$, allowing the quantile function to vary  both in $\bm X$ and $\tau$.  Indeed, this work goes beyond existing approaches that assume that $\tau$ is fixed \citep{kindo2016bayesianquantileadditiveregression}, aiming to estimate  the entire quantile function for any $\bm X\in\mathcal X$. The IQ-BART model is a variant of Bayesian non-parametric constructions for density regression \citep{orlandi2021densityregression,li2023adaptive}. However, instead of targeting the conditional distribution directly (e.g. through covariate-dependent mixtures), IQ-BART assigns a BART prior over the quantile function.  Note that the shape of the conditional distribution $Y$ given $\bm X$ is determined by the shape of the Laplace location parameter $\mu=f_\theta(\bm X,\tau)$ for $\tau\in(0,1)$ as opposed to the Laplace distribution itself.

Data augmentation with $\tau_i$'s  and the asymmetric Laplace likelihood with a forest location parameter that depends on both $\bm x_i$'s and $\tau_i$'s give rise to Implicit Quantile BART (IQ-BART).  The aim of this new method is to posit a Bayesian model on the conditional quantile function, acting as a model over a conditional model for $Y_i$ given $\bm x_i$. The posterior distribution over $\mu=f_\theta(\bm X,\tau)$ can be then distilled for estimation of the {\em entire quantile function} as well as for assessing uncertainty through the variation of posterior draws. Similarly to \citet{shen2024distributional}, our approach is ultimately generative where sampling from the conditional distribution can be accomplished through an estimate of the quantile function yielding to sampling-based estimation for the conditional mean and quantiles.
Uncertainty quantification is a unique feature of our approach that is not available for  related optimization-based quantile learning methods. Related approaches such as \citet{shen2024distributional} could be equipped with uncertainty quantification using bootstrap techniques, see e.g. \citet{nie2023deep} for an overview.

The fully Bayesian framework allows us to approximate the posterior distribution of the entire quantile function through MCMC simulation.
The approximate posterior immediately yields the posterior mean quantile function estimator which is in principle different from the posterior mode estimator which would correspond to the classical (sample) quantile estimator. We assess the discrepancy of the posterior mode and mean in unconditional quantile estimation and show that the posterior mean exhibits automatic interpolation and more often favorable $L_{2}$ risk properties in finite samples. For the full-blown IQ-BART  method, we provide theoretical support by quantifying the rate at which the posterior distribution concentrates around the true conditional quantile function. Finally, we demonstrate the efficacy of the approach on synthetic distributional regression examples, generic time-series one-step-ahead forecasting, and modeling predictive distributions for one-quarter-ahead economic and financial conditions. The latter example shows that time series may exhibit bimodal conditional distributions which make the data not amenable to classical parametric analysis. It is precisely with flexible techniques such as IQ-BART that one can fully capture the frame of these distributions without stringent distributional assumptions.

The paper is structured as follows. Section~\ref{sec:bayes-quantile-fn-est} begins by reviewing Bayesian estimation of the quantile from the misspecific likelihood and loss-based perspectives and includes a study of the Gibbs posterior mean as an estimator. In Section~\ref{sec:iqbart} we introduce IQ-BART and discuss variations involving further data augmentation via repeated quantiles as well as a multivariate extension. Section~\ref{sec:theory} provides a theoretical guarantee on the performance of IQ-BART via an upper bound on the rate of posterior contraction around the true conditional quantile function. Section~\ref{sec:sampling} discusses an approximate posterior sampling algorithm for sampling from the IQ-BART posterior and touches on sampling from predictive distributions given covariate values. We perform a simulation study in Section~\ref{sec:simulation-study}, benchmarking IQ-BART relative to related distributional regression techniques on three example data-generating processes. Finally, we apply IQ-BART for one-step-ahead time-series forecasting both on U.S. economic and financial conditions as well as on datasets from the M4 Forecasting competition \citep{makridakis2020m4}. We provide our concluding remarks in Section~\ref{sec:conclusion}.

\section{Bayesian Quantile Learning}\label{sec:bayes-quantile-fn-est}
For $\tau\in (0,1)$, the $\tau$-quantile of a probability measure $P$  is defined as any value $q_{\tau}$ that satisfies $\P_{Y\sim P}(Y\leq q_{\tau})=\tau$.
It may alternatively be characterized as a population risk minimizer for the $\tau$-check loss function $\rho_{\tau}$ defined via $\rho_{\tau}(x)=x(\tau-\mathbb{I}\{x<0\})$, i.e.
\begin{equation}
q_{\tau}\in \argmin_{\beta} \E_{Y\sim P} \left[ \rho_{\tau}(Y-\beta) \right].\label{eq:quantile-defn}
\end{equation}
The quantile $q_{\tau}$ is unique when the distribution function of $P$ is strictly increasing and continuous. In such a case, the distribution function $F$ is invertible and we may write $q_{\tau}=F^{-1}(\tau)$.

A natural $\tau$-quantile estimator from $iid$ realizations $\bm Y=(Y_1,\dots, Y_n)'$ is the sample quantile $\hat{q}_{\tau}$ which can be cast as a solution to the following empirical version of the optimization problem
$$
\hat{q}_{\tau}\in\argmin_{\beta}\sum_{i=1}^n\rho_\tau(Y_i-\beta).
$$
Instead of specifying the estimand through the loss function, $\hat{q}_{\tau}$ can alternatively be regarded as an estimator of a parameter in a certain parametric model.
 It has been common practice for Bayesian practitioners~\citep{tsionas2003bayesianquantileinference, yu2001bayesianquantileregression, kindo2016bayesianquantileadditiveregression} to approach Bayesian quantile learning by employing an intentionally mis-specified  auxiliary model using  asymmetric Laplace likelihood on  $Y_i$'s (defined below).

\begin{defn}
  The \emph{asymmetric Laplace distribution} $\mathrm{ALD}(\mu,\lambda,\tau)$ is a   probability measure  on the real line with density
  \begin{equation}\label{eq:asym-lap-dens}
  \pi(y \C \mu, \lambda, \tau)=\lambda^{-1}\tau(1-\tau)\exp\{-\rho_\tau(y-\mu)/\lambda\}
  \end{equation}
  parameterized by location $\mu\in \mathbb{R}$, scale $\lambda>0$, and tilt $\tau\in (0,1)$.
\end{defn}

Unsurprisingly, $\hat{q}_{\tau}$ can be regarded as the maximum likelihood estimator of the location parameter $\mu$ of the asymmetric Laplace distribution from iid observations. Under the flat prior $\pi(\mu)\propto 1$, this estimator coincides with the {\em posterior mode}.
Contrastingly, the {\em posterior mean} estimator under the Laplace likelihood and the flat prior is given by
\begin{align}\label{eq:quantile-post-mean}
\hat{q}_{\tau,\lambda}^{\mathrm{PM}}\equiv \E \left[ \mu\mid Y_{1},\ldots,Y_{n} \right]&=\frac{\int \mu \prod_{i=1}^n \pi(Y_i\C \mu,\lambda,\tau)\d \mu}{\int \prod_{i=1}^n \pi(Y_i\C \mu,\lambda,\tau)\d \mu}\\ &=\frac{\int\mu \exp\left\{-\lambda^{-1}\sum_{i=1}^n\rho_\tau(Y_i-\mu)\right\}\d\mu}{\int \exp\left\{-\lambda^{-1}\sum_{i=1}^n\rho_\tau(Y_i-\mu)\right\}\d\mu}\,.
\end{align}
Before investigating properties of the estimator~\eqref{eq:quantile-post-mean}, we foray into an alternative viewpoint of Bayesian quantile estimation.

Both estimators $\hat{q}_{\tau}$ and $\hat{q}_{\tau,\lambda}^{\mathrm{PM}}$ have Bayesian interpretations as the posterior mode and mean, respectively, for the location parameter $\mu$ of a $\tau$-tilted asymmetric Laplace likelihood with a scale parameter $\lambda$.
Under the misspecified model, the $\hat{q}_{\tau}$ and $\hat{q}_{\tau}^{\mathrm{PM}}$ would be Bayes estimators for the $0$-$1$ loss and squared error loss, respectively, under the improper flat prior. \citet{yu2001bayesianquantileregression} demonstrate that the flat improper prior on the location parameter does indeed lead to a proper posterior.
The sample quantile (posterior mode) estimator $\hat{q}_{\tau}$  is asymptotically unbiased for the $\tau$ quantile, and furthermore, if $P$ has a continuous and monotone increasing distribution function $F$ with density $f$, then the $\tau$-sample quantile $\hat{q}_{\tau}$ {converges in distribution to a Gaussian centered around the true $\tau$-quantile with variance matching that in Lemma~\ref{lem:posterior_mean_convergence} \citep{ruppert2015statistics}.}
The posterior mean, however, has been somewhat overlooked in the literature. We attempt to perform a comparative analysis of these two estimators.

An alternative way of viewing the actual posterior under the misspecified Laplace model is through  the \emph{Gibbs posterior} \citep{zhang2006from, zhang2006information, jiang2008gibbs, bissiri2016general, martin22_direc_gibbs_poster_infer_risk_minim}.
This generalized Bayesian approach allowing direct inference on a risk minimizer rather than an unknown parameter linked to the data via a likelihood \citep{bissiri2016general, martin22_direc_gibbs_poster_infer_risk_minim}. Indeed, the characterization of the $\tau$-quantile as a risk minimizer derived from the check-loss  allows for direct uncertainty quantification over the quantile via the Gibbs posterior distribution.
Some loss functions may yield Gibbs posteriors which correspond to actual posteriors with a recognizable loss-based likelihood (our case here).
With these considerations in mind, we note that $\hat{q}_{\tau}$ and $\hat{q}_{\tau,\lambda}^{\mathrm{PM}}$ are the mode and mean of the Gibbs posterior for the loss $\ell_{\tau}$ with \emph{learning rate} $\lambda^{-1}$. Under mild conditions on the prior, properties of the Gibbs posterior for the $\tau$-quantile include consistency for $q_{\tau}$ and asymptotic normality \citep{chernozhukov2003classical, martin22_direc_gibbs_poster_infer_risk_minim}.

A practitioner may wish to avoid using Gibbs posteriors (or Laplace-induced posteriors) via nonparametric estimation.  Since the quantile defined through loss minimization, it is possible to proceed by assigning a non-parametric prior on the distribution function to obtain uncertainty quantification for the quantile estimate.
Some of these approaches include Dirichlet process priors \citep{kottas2009bayesian, hjort2007nonparametric, taddy2010bayesian} and  quantile pyramid priors \citep{hjort2009quantile}. Posterior inference over a particular quantile is then available as functional of the posterior probability measure.

\subsection{The Posterior Mean Quantile Estimator}

Our  IQ-BART procedure is implemented via MCMC whose samples can be quickly summarized by the posterior mean (as opposed to  the mode)  quantile estimator. The mean estimator $q^{\mathrm{PM}}_{\tau,\lambda}$ lends itself to a relatively palatable closed-form solution in i.i.d. setting. While the mode estimator has been used widely, we could not find any characterization and investigation of the posterior mean as a quantile estimator in the literature. We attempt to provide some more context below.

\begin{lemma}\label{lem:flat-pm}
  Given a univariate dataset $Y_{1},\ldots,Y_{n}$, let $Y_{(k)}$ denote the $k$th order statistic, with $Y_{0}\equiv -\infty$ and $Y_{n+1}\equiv \infty$. For $0\leq k \leq n$, define
  \[
    w_{k}\equiv \exp\left(-\lambda^{-1}(\tau-1) \sum_{i=1}^{k}Y_{(i)}-\lambda^{-1}\tau\sum_{i=k+1}^{n}Y_{(i)}\right)\,,
  \]
  as well as
\begin{align*}
  \psi_{k} &\equiv  \begin{cases}\frac{1}{c_{k}^{2}}\left[(c_{k}Y_{(k)}+1)\exp(-c_{k}Y_{(k)})-(c_{k}Y_{(k+1)}+1)\exp(-c_{k}Y_{(k+1)})\right] & k/n \neq \tau \\ \frac{1}{2}\left[Y_{(k+1)}^{2}-Y_{(k)}^{2}\right] & k/n = \tau\end{cases} \\
  \varphi_{k} &\equiv  \begin{cases}\frac{1}{c_{k}}\left[\exp(-c_{k}Y_{(k)})-\exp(-c_{k}(Y_{(k+1)})\right] & k/n \neq \tau \\ \left[Y_{(k+1)}-Y_{(k)}\right] & k/n = \tau\end{cases}\,.
\end{align*}
where $c_{k}:= \lambda^{-1}( k-n\tau)$. Then, the posterior mean estimator $\hat{q}_{\tau,\lambda}^{\mathrm{PM}}$ can be computed via the ratio of inner products
\begin{equation}\label{eq:pm-form}
  \hat{q}_{\tau,\lambda}^{\mathrm{PM}}=\frac{\bm{w}^{\top}\bm{\psi}}{\bm{w}^{\top}\bm{\varphi}}\,.
\end{equation}

\end{lemma}

\begin{proof}
  Refer to Section~\ref{sec:pf-lem-1}.
\end{proof}

{Regardless of the learning rate $\lambda^{-1}$, the posterior mean estimator $\hat{q}_{\tau,\lambda}^{\mathrm{PM}}$ converges to the sample quantile $\hat{q}_{\tau}$ as the sample size grows large.

\begin{lemma}\label{lem:posterior_mean_convergence}
Suppose $n$ iid observations are drawn from a probability measure with strictly increasing and continuous distribution function $F$ and density $f$. Under the conditions of Theorem 3 of \citep{martin22_direc_gibbs_poster_infer_risk_minim}, for any learning rate $\lambda^{-1}>0$, the posterior mean estimator satisfies
\[
\hat{q}_{\tau,\lambda}^{\mathrm{PM}}\overset{D}{\to}\mathcal{N}\left(F^{-1}(\tau),\frac{\tau(1-\tau)}{n\left(f\left(F^{-1}(\tau)\right)\right)^{2}}\right)
\]
where the convergence is in distribution as the sample size $n$ grows.
\end{lemma}

\begin{proof}
  Refer to Section~\ref{subsec:asymp-pf}.
\end{proof}

Lemma~\ref{lem:posterior_mean_convergence} indicates that the asymptotic behavior of the posterior mean estimator matches that of the sample quantile regardless of $\lambda$. Any advantage of the posterior mean estimator arrives in its ability to estimate quantiles in small sample sizes, where it may have favorable interpolation qualities when compared to the sample quantile (see Figure \ref{fig:risk-comp} for risk comparisons). We will now switch our focus to studying the effect of varying $\lambda$ on the finite sample behavior of the posterior mean estimator.
}

\subsection{Calibration of the Learning Rate}

\begin{figure}[ht]
  \centering
  \includegraphics[width=\linewidth]{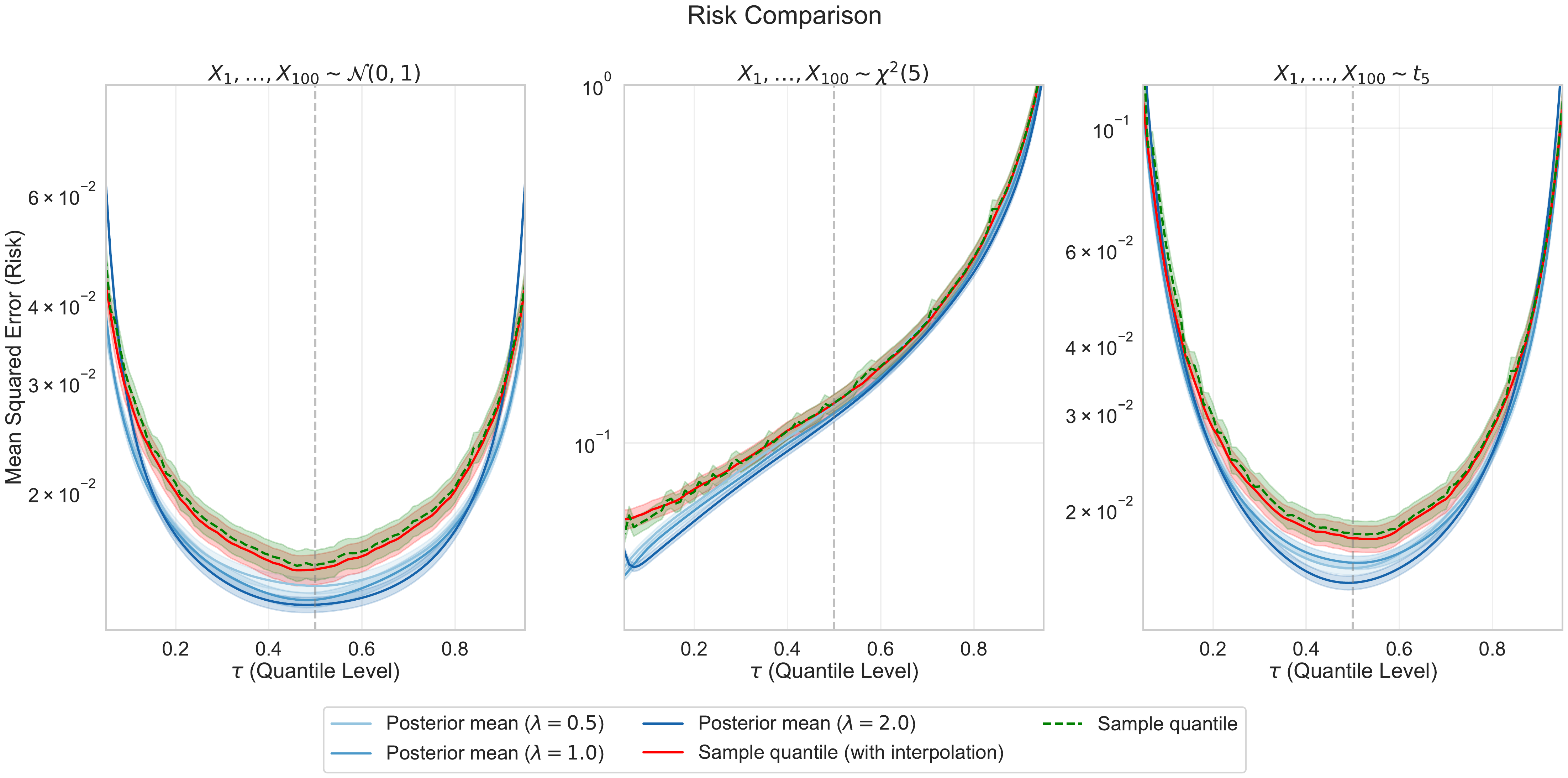}
  \caption{\label{fig:risk-comp} $L_{2}$ risk {(squared error loss)}  profile comparison between the sample quantile and posterior mean (flat prior) in estimation of the $\tau$ quantile from 100 $iid$ samples of three example distributions. Risk is estimated from the mean of 10,000 data replications. Shaded regions indicate one standard error. We also assume the support of each probability measure is not known at the time of estimation. }
\end{figure}

When employed as an estimator for the risk minimizer of the check-loss induced risk of general probability measures, eliciting the learning rate $\lambda^{-1}$ becomes a salient concern for the purpose of efficient estimation or well-calibrated posterior credible sets. As $\lambda$ approaches 0, the posterior mean estimator converges to the sample quantile. This can be clearly seen through the Gibbs posterior framework, as the learning rate $\lambda^{-1}$ converging to infinity implies that the likelihood dominates and the prior becomes irrelevant. In this case, the Gibbs posterior mean converges to the empirical risk minimizer (the sample quantile). On the other hand, when $\lambda^{-1}$ converges to zero, the prior dominates and the Gibbs posterior mean typically converges to the mean of the prior. In the case where the support is unbounded, the flat prior is improper and does not have have a finite mean, which can lead to unclear behavior.

{ 

\begin{lemma}\label{lem:asymp-lam}
  As the learning rate satisfies $\lambda^{-1}\to \infty$, the posterior mean satisfies $\hat{q}_{\tau,\lambda}^{\mathrm{PM}}\to \hat{q}_{\tau}$. As the learning rate satisfies $\lambda^{-1}\to 0$, the posterior mean resulting from the improper prior satisfies $\hat{q}^{\mathrm{PM}}_{0.5,\lambda}\to (1/n)\sum_{i=1}^{n}{Y_{i}}$ and diverges for $\tau\neq 0.5$. When the uniform prior is no longer improper but proper and truncated in the interval $[a,b]$, the posterior mean estimator satisfies $\hat{q}^{\mathrm{PM}}_{\tau,\lambda}\to (a+b)/2$.
\end{lemma}

\begin{proof}
  Refer to Section~\ref{subsec:asymp-pf-lam}.
\end{proof}

}

Practitioners may use a variety of tools for calibration of $\lambda^{-1}$. Practitioners may optimize for predictive accuracy by tuning $\lambda$ on out-of-sample data. The general posterior calibration algorithm of~\citet{syring2018calibrating} can alternatively be employed such that credible intervals are properly calibrated in the frequentist sense. Figure~\ref{fig:risk-comp} shows the simulated risk (induced by squared error loss) taken over 10,000 data replications in estimating the $\tau$ quantile of 100 $iid$ samples from three distributions: standard Gaussian, chi-square with five degrees of freedom, and a Student $t$ distribution with five degrees of freedom. We observe that the the posterior mean can often have lower risk for various quantile values when compared to the sample quantile.

The posterior mean estimator can provide compelling advantages over the sample quantile as demonstrated by the risk profiles across various distributions in Figure~\ref{fig:risk-comp}, which shows empirical approximations of the $L_{2}$ risk in estimating quantile values $\tau\in (0,1)$ from $n=100$ samples for a Gaussian, a chi-squared, and a Student $t$ distribution respectively. The risk profile of the posterior mean estimator is favorable when compared to the sample quantile, with lower risk for most quantile values among these three examples. In the next section, we presume with our novel prior distribution over the conditional quantile function that may be used in a Bayesian analysis. We use these insights to help justify the usage of the posterior mean as a desirable point estimator of the conditional quantile function, as it is computationally trivial to compute from samples while maintaining good performance in the unconditional setting.

\section{Implicit Quantile BART}\label{sec:iqbart}
Our proposed Implicit Quantile BART (IQ-BART) method
aims to perform statistical inference and uncertainty quantification for the conditional model for $Y$ given $\bm{X}$ through the conditional quantile function $Q(\tau \C \bm{X}=\bm{x})=F^{-1}(\tau\C \bm{X}=\bm{x})$ (for $F$ invertible) for   $\tau\in(0,1)$ based on observed covariate vectors and responses $\{(Y_i,\bm X_i)\}_{i=1}^n\iid P$.

For a fully Bayesian analysis, we  exploit the  connection between quantile estimation  and maximum likelihood estimation under an asymmetric Laplace distribution mentioned previously in Section \ref{sec:bayes-quantile-fn-est}.  Recall that the purpose of this likelihood is not necessarily to attempt to model the data, but rather to identify an estimation target $\mu(\bm{X},\tau)$, i.e. the quantile function of the unknown distribution of $Y$ given $\bm{X}$. The distributional assumption for $Y$ given $\bm{X}$ is not expressed in the Laplace likelihood. It is baked inside structural assumptions on the location parameter $\mu(\bm{X},\tau)$ as a function of $\tau$ for fixed $\bm{X}$. This quantile function completely characterizes the distribution of $Y$ given $\bm{X}$.
Assigning a prior over $\mu(\bm{X},\tau)$ (for given $\bm{X}$) is thus equivalent to assigning a prior distribution over conditional distributions for $Y$ given $\bm{X}$.

If we were interested in one pre-specified quantile $\tau\in (0,1)$, we could proceed with a  Bayesian analysis where all observations $\{(Y_i,\bm X_i)\}_{i=1}^n$  would be used to estimate the single conditional $\tau$ quantile.  This would mean that all observations $Y_i$ would share one $\tau$ in the asymmetric Laplace likelihood specification in~\eqref{eq:asym-lap-dens}.
This strategy was considered by \citet{kindo2016bayesianquantileadditiveregression} who imposed a forest structure on the location parameter.
Our goal is more ambitious, however, because we want to estimate the entire quantile function $\mu(\bm X,\tau)$ for all $0<\tau<1$. We achieve this task through a horizontal data augmentation, i.e. adding one additional covariate.

In its simplest form, our approach consists of  generating  an auxiliary random variable  $\tau_i\in (0,1)$ and attaching it to the observation $(Y_i,\bm X_i)$ as a training point for learning the
$\tau_i^{th}$ conditional quantile. We discuss various data augmentation strategies later in Section \ref{sec:data_aug} where multiple random quantiles are generated for each observation $(Y_i,\bm X_i)$. Due to data-augmentation, we are generating fake quantile values that we want to extrapolate from  to learn the entire quantile function for all $\tau$ values.
Endowing the $i^{th}$ row of the training data with a uniformly (independently) sampled random variable $\tau_{i}$, we construct the auxiliary observation model for the responses from asymmetric Laplace distribution, but now with different $\tau_i$ for each observation, i.e.
\begin{equation}
Y_{i}\sim \mathrm{ALD}(\mu_{i}, \lambda, \tau_{i}).\label{eq:ALD}
\end{equation}
The location parameter $\mu_i$ varies from datapoint to datapoint, depending on the values $\bm{X}_i$ as well as $\tau_i$.
In order to borrow strength from ``nearby" observations $(Y_j,\bm X_j)$ (with $\tau_j$ that are similar to $\tau_i$) for the estimation of $\tau_i^{th}$ quantile, it is desirable that  $\mu_i$'s vary somewhat smoothly  in $\tau$.
Now, because $\tau_i$ are different for each datapoint, we have only one observation $Y_i$ to learn about the quantile function $Q(\tau_i \C \bm{X}_i)$.
Relating $\mu_i$ to $\bm{X}_i$ through a forest aggregator enables to borrow strength from observations $Y_j$ for which $\bm{X}_j$ is similar to $\bm{X}_i$.
To enhance information sharing among the observations even more, we allow $\mu_i$ to depend on $\bm X_i$ as well. Namely, we have
\begin{equation}
\mu_i=F_T( \bm{X}_i,\tau_i)=\sum_{k=1}^T \mathcal T_k( \bm{X},\tau)\label{eq:mu}
\end{equation}
where
\begin{equation}
 \mathcal T_k\iid \pi(\mathcal T)\label{eq:tree_prior}
\end{equation}
is a forest of $T$ random binary trees $T_k$ that split on both $\bm X$ as well as $\tau$ and arrive from the Bayesian CART prior $\pi(\mathcal T)$ \citep{chipman1997bayesian, chipman2010bart}.
In other words, the newly manufactured simulated values $\tau_i$ for $1\leq i\leq n$ enter the model in two ways: (1) as a new predictor which allows the trees to split on it to borrow strength from  $\tau_j$ values that are similar to $\tau_i$ when estimating a quantile $Q(\tau_i \C \bm{X})$, (2) as a tilt in the contribution of $Y_i$ to the overall likelihood function.

If one were to estimate $\mu_i$ purely using a frequentist method, one could deploy some parametrization $\mu_i=g_\theta(\bm{X}_{i},\tau_{i})$ and search for
 $$
\hat\theta=\arg\min_{\theta}\sum_{i=1}^n\rho_{\tau_i}(Y_i-g_\theta(\bm{X}_i,\tau_{i})).
$$
This was done by \citet{dabney2018implicit} using deep learning mappings $g_\theta$ and later exploited by \citet{polson23_gener_ai_bayes_comput}. Learning these mappings can be challenging. In addition, imposing monotonicity in $\tau_i$ of the mapping $g(\cdot)$ is not trivial. Obtaining a proper quantile function requires a specialized architecture or optimization technique, or a post-minimization nondecreasing rearrangement. Lastly, one obtains only a point estimator $g_{\hat\theta}(\cdot)$ of the quantile function without any uncertainty quantification. We adopt a fully Bayesian approach using forests as opposed to deep learning.

The prior \eqref{eq:tree_prior} expresses itself as a regularizer in the posterior, obtaining smoother and more robust results than deep learning (as we will show later). The IQ-BART posterior over the conditional quantile function values $\mu$ takes the form
\begin{equation}
\pi(\mu\mid \bm{Y},{\bf X},\bm{\tau})\propto \exp\left\{-\lambda^{-1}\sum_{i=1}^{n}\rho_{\tau_{i}}(Y_{i}-\mu(\bm{X}_{i},\tau_{i}))\right\}\pi(\mu)\,.\label{eq:posterior}
\end{equation}
where $\bm Y=(Y_1,\dots, Y_n)',\bm\tau=(\tau_1,\dots, \tau_n)'$ and ${\bf X}=[\bm X_1',\dots, \bm X_n']'$.
It is more adequate to regard this posterior as the  ``Gibbs posterior" under the check loss. While this posterior does have an interpretation as a genuine posterior under an asymmetric Laplace likelihood, it serves as a vehicle for inference about our estimand (quantile function) as opposed to an underlying parameter value which would give rise to the observed data under the postulated model.

 In addition, we are able to impose a monotonicity constraint using  the monotone BART variant \citep{chipman2022mbart}.
 Finally, the posterior draws $F_T^{(1)},\dots, F_T^{(M)}$ of  $F_T(\cdot)$ as a function of both $\tau$ and $\bm{X}$ allow us to simulate posterior distributions over the conditional quantile function
 $Q(\tau \C \bm{X}) $ for any fixed value $\tau$ and/or $\bm{X}$. We study the frequentist properties of the posterior in the context of quantile function estimation in Section~\ref{sec:theory}. One can distill the  $M$ posterior draws into the Bayesian quantile function estimator
 \begin{equation}\label{eq:cond-post-mean-estimate}
 \hat Q(\tau \C \bm{X}=\bm{x})=\frac{1}{M}\sum_{m=1}^MF_T^{(m)}(\bm{x},\tau).
 \end{equation}
 which is a conditional analogue to the posterior mean estimator described and studied in Section~\ref{sec:bayes-quantile-fn-est}  under a regularization prior.

\subsection{Data Augmentation Strategies}\label{sec:data_aug}
Our implementation of data augmentation using auxiliary uniform random variables $\tau_i\in (0,1)$ consists of sampling an additional covariate and considering it as fixed. Another approach would be to regard them as latent variables that are a part of the hierarchical model and refresh them using simulation.

This idea was implemented in a related approach by \citet{orlandi2021densityregression} who, building on \citep{kundu2014latent}, assume $Y_i=\mu(\bm{X}_i,\tau_i)+\varepsilon_i$ for $\varepsilon_i\iid N(0,\sigma_i^2)$ and, independently, $\tau_i\iid U[0,1]$. Integrating out $\tau_i$'s, this induces a marginal Gaussian mixture model for $Y_i$, given $\bm{X}_i$. Instead, we directly assume that  $Y_i=\mu(\bm{X}_i,\tau_i)$ for some $\tau_i\in (0,1)$ for a push-forward (quantile) map $\mu(\bm{X}_i,\cdot)$ that transforms uniformly sampled noise onto the observed values. The goal is to estimate this map from $\{(Y_i,\bm{X}_i)\}_{i=1}^n$ as opposed to build a generative model that gave rise to $\{(Y_i,\bm{X}_i)\}_{i=1}^n$.

We now introduce a unifying perspective on the IQ-BART posterior \eqref{eq:posterior} through the lens of  Gibbs posteriors under the check loss.  Instead of attaching one single value $\tau_i$ to  $(\bm X_i, Y_i)$ (as in \eqref{eq:posterior}), one could consider the integrated check-loss
\begin{equation}\label{eq:int-risk}
\ell(\mu;\bm{X}_i, Y_i) = \int_{0}^{1}\rho_{\tau}(Y_{i}-\mu(\bm{X}_{i},\tau))\,\mathrm{d}\tau
\end{equation}
which was recently analyzed in the unconditional setting for quantile-function estimation~\citep{narayan2024expected}. It is clear that given  true the conditional quantile function $\mu_{0}=Q(\tau \C \bm{X}=\bm{x})$ that gave rise to $\{(Y_i,\bm X_i)\}$, we have $\mu_{0}\in\argmin_{\mu} \mathbb{E}_{\bm{X},Y\sim P} \left[ \ell(\mu;\bm{X},Y) \right]$. The Gibbs posterior with respect to this loss function with a learning rate $\lambda^{-1}$ is then given by
\begin{align*}
  \pi(\mu\mid \bm{Y},\bm{X},\bm{\tau})&\propto \exp\left\{-\lambda^{-1}\sum_{i=1}^{n}\ell(\mu;\bm{X}_{i},Y_{i})\right\}\pi(\mu) \\
  &= \exp\left\{-\lambda^{-1}\sum_{i=1}^{n}\int_{0}^{1}\rho_{\tau}(Y_{i}-\mu(\bm{X}_{i},\tau))\,\mathrm{d}\tau\right\}\pi(\mu)\,.
\end{align*}
Indeed, such a Gibbs posterior targets the minimizer of the risk induced by~\eqref{eq:int-risk} as an inferential target.

\sloppy Instead of assigning a fixed, uniformly sampled quantile value $\tau_{i}$ to each observation $(\bm{X}_{i}, Y_{i})$, we can further augment the data by considering repeated observations $(\bm{X}_{i}^{(1)},Y_{i}^{(1)}),\ldots,(\bm{X}_{i}^{(r)},Y_{i}^{(r)})$ and augmenting each copied dataset indexed by $j=1,\ldots,r$ by a independently sampled $\tau_{j}\sim\mathcal{U}(0,1)$. Such a scheme resembles a simultaneous quantile regression over the $r$ quantile values $\tau_{1},\ldots,\tau_{r}$ with posterior density
\begin{align}\label{eq:iqbart-density-simultaneous}
  \pi(\mu\mid \bm{Y},\bm{X},\bm{\tau})\propto \exp\left\{-\lambda^{-1}\sum_{i=1}^{n}\sum_{j=1}^{r}\rho_{\tau_{j}}(Y_{i}^{(j)}-\mu(\bm{X}_{i}^{(j)},\tau_{j}))\right\}\pi(\mu)\,.
\end{align}
Note that such a procedure replaces the loss~\eqref{eq:int-risk} in the generalized likelihood factor by its Monte Carlo approximation by $r$ samples. Despite no data occurring outside the fixed set of $r$ quantile values, the sum-of-trees prior over $\mu$ can still performantly interpolate on other quantile values. We study the theoretical properties of the posterior distribution described in~\eqref{eq:iqbart-density-simultaneous} in Section~\ref{sec:theory}.

Finally, a similar approach entails independently sampling separate $\tau_{i}^{(1)},\ldots,\tau_{i}^{(r)}$ for $i\in[n],j\in[r]$ to pair with each identical copy $(\bm{X}_{i}^{(j)},Y_{i}^{(j)})$. This also allows individual data points to be re-used for the learning of multiple quantile values simultaneously while avoiding the potential pitfall of the set of modeled quantiles being only of size $r$. In this case, using the augmented reference table $\{\bm{X}_{i}^{(j)},Y_{i}^{(j)},\tau_{i}^{(j)}:i\in[n],j\in[r]\}$, the IQ-BART posterior over $\mu$ takes the form
\begin{align}\label{eq:iqbart-density-fully-augmented}
  \pi(\mu\mid \bm{Y},\bm{X},\bm{\tau})\propto \exp\left\{-\lambda^{-1}\sum_{i=1}^{n}\sum_{j=1}^{r}\rho_{\tau_{i}^{(j)}}(Y_{i}^{(j)}-\mu(\bm{X}_{i}^{(j)},\tau_{i}^{(j)}))\right\}\pi(\mu)\,.
\end{align}
These approaches imply that we can interpret the deterministic assignment of $\tau_{j}$ (or $\tau_{i}^{(j)}$) values to each data point as a Monte Carlo approximation of the loss~\eqref{eq:int-risk}. Repeating multiple quantile values for each data point increases the fidelity of this Monte Carlo approximation, but comes at the cost of increased computational complexity in approximate posterior sampling. The additional cost of using an augmentation of size $r$ is equivalent to the increase in cost if the number of observations in standard BART sampling has increased from $n$ to $nr$.

{In our work, we primarily consider a data augmentation strategy through a fixed augmentation of uniformly sampled $\tau$ values occurring prior to approximate posterior inference. Other strategies could include refreshing the $\tau$ values to be different uniform samples through the Gibbs sampling procedure, which would allocate data points to different quantiles for different posterior draws rather than for the same one each time.}

{ 
\subsection{Multivariate IQ-BART}

The previous discussions have focused on modeling the conditional distribution of a univariate response $Y$ given a $d$-dimensional covariate vector $\bm{X}$. Recently, \citet{cevid2022journal} proposed a forest-based distributional regression method that may be applied to a multivariate response $\bm{Y}$. \citet{kim2025deepgenerativequantilebayes} extended the approach to generative Bayesian computation via implicit quantile regression of \citet{polson23_gener_ai_bayes_comput} to the multivariate case by parameterizing potential functions in the dual objective of the Kantorovich problem.

Indeed, the existence of the quantile function as a pushforward measure from a uniform distribution is a uniquely univariate phenomenon. An extension of IQ-BART to the multivariate response case may involve a prior over a conditional transport map $T:\mathcal{X}\times\mathcal{V}\to \mathcal{Y}$, where $\mathcal{V}$ is the support of a reference distribution, e.g. $[0,1]^{d}$. The analogue of the conditional quantile function may become the conditional Monge map which minimizes the transport cost over maps that push the reference distribution to the conditional probability measure of the response (or any other valid conditional transport map). An extension may require specifying a prior over step functions constituting valid transport maps and constructing a loss-based likelihood based on the transport cost for Gibbs posterior, which raises concerns of intractability. Aiming for the Brenier map by parameterizing the convex potential is not amenable to step-function style priors due to their lack of differentiability.

Our extension to the multivariate case avoids these pitfalls by targeting the Knothe-Rosenblatt map, a composition of conditional quantile transforms that include increasingly more response components. \citet{carlier2008knothestransportbreniersmap} show that the Knothe-Rosenblatt map is the limiting solution to a sequence of Monge-Kantorovich problems with quadratic costs. While giving up optimality (in the transport cost sense) and adding dependence on the covariate ordering, the Knothe-Rosenblatt map is advantageous in its ability to be efficiently modeled with step functions. This map can be expressed as the recursive application of conditional quantile functions that transport $(\tau_{1},\ldots,\tau_{k})\sim \otimes_{j=1}^{k}\mathcal{U}[0,1]$ to $\bm{Y}\mid \bm{X}$ via
\begin{align*}
  Y_1 &= Q_{Y_1 | X=x}(\tau_1) \\
  Y_j &= Q_{Y_j | Y_1, \ldots, Y_{j-1}, X=x}(\tau_j) \quad\text{for $j=1,\ldots,k$}\,.
\end{align*}

\begin{figure}[ht]
  \centering
  \includegraphics[width=0.6\linewidth]{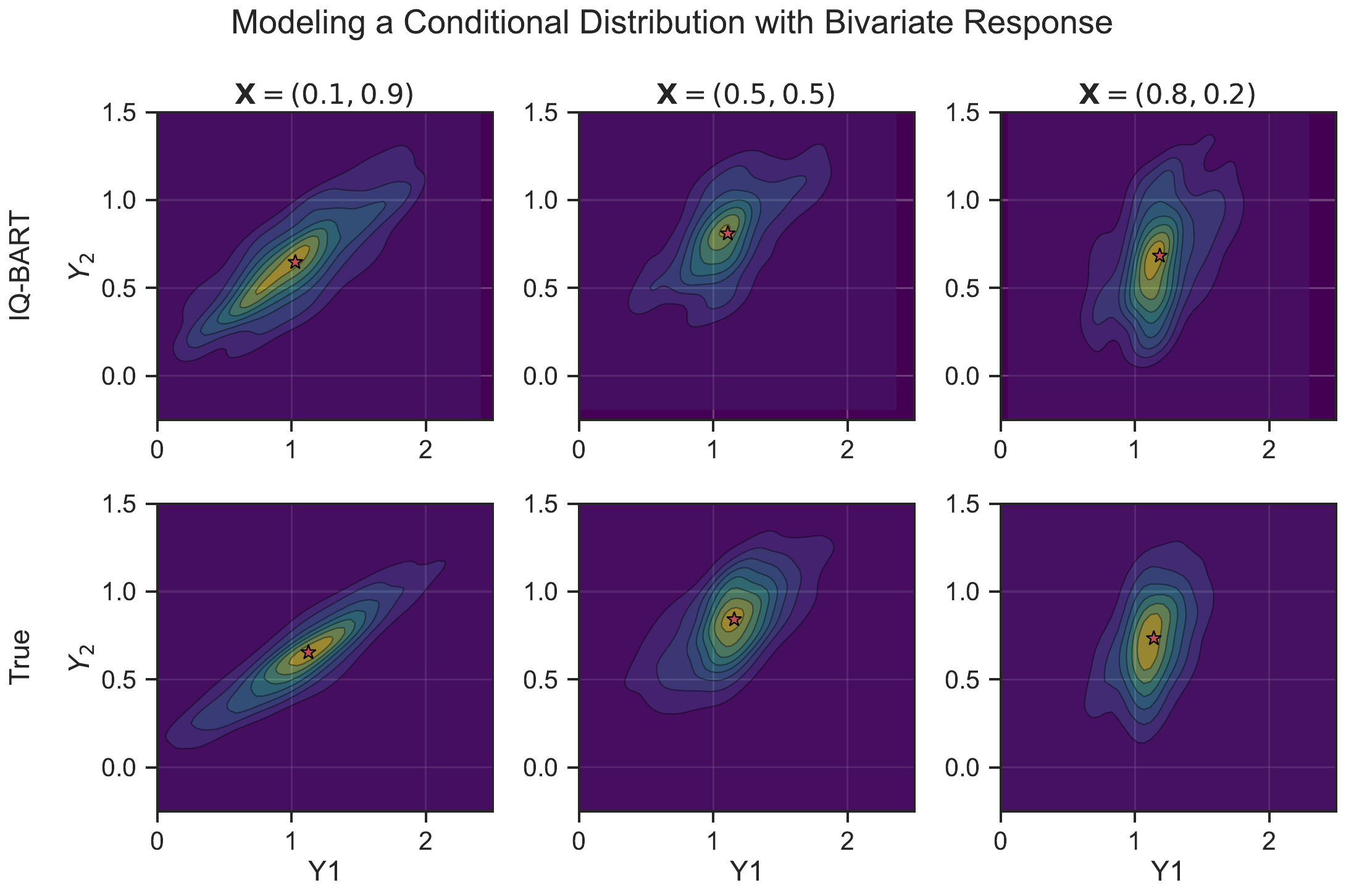}
  \caption{\label{fig:bivariate-iqb} Demonstration of multivariate IQ-BART on a synthetic data example. We fit IQ-BART on 1000 data-points from data generated by \eqref{eq:cond-dgp-bivar}. We plot the resulting conditional distributions on the response for three values of covariates.}
\end{figure}

The distributional regression problem can be posited in a statistical sense by considered this sequence of conditional quantile functions as the parameter, and endowing it with a prior distribution. We endow the $j$th conditional quantile function with a BART prior supported on step functions  $\mathbb{R}^{d+j-1}\times [0,1]\to \mathbb{R}$ for $j=1,\ldots,k$. Analogously to the univariate case, we now augment each data point indexed by $i$ with a $k$-tuple of uniformly sampled quantiles $\bm{\tau}_{i}:=(\tau_{i1},\ldots,\tau_{ik})$ to form the training data tuple $(\bm{X}_{i},\bm{Y}_{i},\bm{\tau}_{i})\in \mathbb{R}^{d}\times\mathbb{R}^{k}\times [0,1]^{k}$. The likelihood has the form
\[
\pi(\mu_{j}\mid \bm{Y},\bm{X},\bm{\tau})\propto \exp\left\{-\lambda^{-1}\sum_{i=1}^{n} \rho_{\tau_{ij}}(Y_{ij}-\mu_{j}(\bm{X}_{i},Y_{1},\ldots,Y_{j-1},\tau_{ij})) \right\}\,.\footnote{The case with further augmentation of $r$ quantile values for each data point described in Section~\ref{sec:data_aug} may be analogously employed.}
\]
The benefit of this approach is that a Gibbs sampling algorithm for the IQ-BART prior and likelihood can be easily adapted by simply sampling each $\mu_{j}$ in sequence.

We demonstrate the multivariate case by considering a simple data generating process where $X_{j}\sim\mathcal{U}(0,1)$ for $j\in\{1,2\}$ and
\begin{align}
\begin{split}\label{eq:cond-dgp-bivar}
Y_1 | X &\sim (1-w) \mathcal{N}(X_1 + X_2, (0.2 + 0.3X_2)^2) + w \text{Exp}(2 + 3X_1) \\
Y_2 | X, Y_1 &\sim \mathcal{N}(0.5Y_1 + X_1X_2, (0.1 + 0.2X_1)^2)
\end{split}
\end{align}
\sloppy where $w = (1 + e^{-(2X_1-1)})^{-1}$. In Figure~\ref{fig:bivariate-iqb} we sample $n=1000$ iid samples from the aforementioned data-generating process to inform the IQ-BART posterior using the previously described procedure. We display the learned conditional distributions on $\bm{Y}\mid\bm{X}=\bm{x}$ for $\bm{x}\in\{(0.1,0.9),(0.5,0.5),(0.8,0.2)\}$. IQ-BART is able to model the conditional response distributional reasonably well at various points in covariate space.

}

\section{Theory for IQ-BART}\label{sec:theory}

The implicit quantile BART procedure specifies the conditional quantile function as its inferential target. Through the generalized Bayesian framework, there are two possible ways to view this target: (1) the risk minimizer of a Monte Carlo estimate of the integrated check-loss risk, and (2) the conditional location parameter of an (misspecified) asymmetric Laplace likelihood. In this section, we adopt the first perspective, and view the IQ-BART posterior as a Gibbs posterior that captures our uncertainty over the risk minimizer induced by a loss function. To establish theoretical guarantees for our procedure, we investigate the rate at which the IQ-BART posterior contracts around the true risk minimizer. This is a fundamental measure that quantifies how rapidly the posterior distribution concentrates around the true estimand as the sample size increases.

We assume that the conditional probability measure $P_{Y\mid \bm{X}}$ is identified by a conditional quantile function $\mu_{0}: \mathcal{X}\times (0,1)\to \mathbb{R}$ that is $\nu$-H\"{o}lder smooth for $\nu>0$. {We use $P^{n}=\bigotimes_{i=1}^{n}P_{\bm{X},Y}$ to represent the joint probability over the data $\mathcal{D}_{n}=\{(\bm{X}_{i},Y_{i}:i=1,\ldots,n\}$, and denote the IQ-BART posterior probability measure as $\Pi_{n}^{r}(\cdot \mid \mathcal{D}_{n})$} corresponding to the density described in~\eqref{eq:iqbart-density-simultaneous}, where $r$ denotes the number of repeated uses of each data point (refer to Section \ref{sec:data_aug}). Ideally, we want to show that $\Pi_{n}^{r}$ contracts around $\mu_{0}$ at a rapid rate. {We impose the following assumptions in order to establish this type of theoretical result.}

\begin{assumption} The following hold
\begin{enumerate}
  \item The covariate space is the $d$ dimensional unit cube $\mathcal{X}=[0,1]^{d}$.
  \item The conditional quantile function $\mu_{0}$ is $\nu$-H\"{o}lder continuous for some $0<\nu\leq 1$, and is square-integrable.
  \item The conditional probability measure $P_{Y\mid \bm{X}}$ admits a continuous density $\pi(y\mid \bm{x})$ for each $\bm{x}\in \mathcal{X}$, and there exist constants $\delta>0$, $\alpha>0$ such that $\pi(y\mid \bm{x})>\alpha$ in the interval $(\mu_{0}(\bm{x},\tau)-\delta, \mu_{0}(\bm{x},\tau)+\delta)$ for all $\bm{x},\tau$.
  \item The parameter space $\mathcal{F}=\{\mu\mid\mu: [0,1]^{d+1}\to \mathbb{R}\}$ is endowed with the prior on step functions $\mathcal{X}\times (0,1)\to \mathbb{R}$ described in Section 4.1 of~\citet{rockova2020posterior}. This prior is then restricted to the set $\mathcal{F}_{n}:\{\mu:\|\mu\|_{\infty}\leq \log^{p}n\}$ for some $p>0$.
\end{enumerate}\label{assn:main}
\end{assumption}

\noindent The first assumption could be straightforwardly generalized to compact subsets of $\mathbb{R}^{d}$. The upper bound of 1 on the H\"{o}lder smoothness is necessary when working directly with trees, but could be circumvented by using priors on more flexible tree-like function classes \citep{linero2018, yee2025scalablepiecewisesmoothingbart}. The third assumption is adapted from a similar assumption in \citep{syring20_gibbs_poster_concen_rates_under} used to show that the check-loss satisfies a sub-exponential type condition which is a sufficient condition for Gibbs posterior contraction in the single conditional quantile estimation setting. The fourth assumption delineates the exact prior that we are studying and restricts the prior to a logarithmically bounded sub-class. Such a restriction is likely not necessary (refer to \citep{kleijn2006misspecification, syring20_gibbs_poster_concen_rates_under}) but simplifies our analysis.

We adopt the data augmentation perspective corresponding to the generalized posterior density~\eqref{eq:iqbart-density-simultaneous} which involves augmenting each covariate and response pair by $r$ quantile values, which we denote with the vector $\bm{\tau}=(\tau_{1},\ldots,\tau_{r})'$. We study the posterior contraction under this generalized likelihood and the prior described in Assumption~\ref{assn:main}.4 under the norm
\[
\|\mu_{0}-\mu\|_{2,\bm{\tau}}:=\frac{1}{r}\sum_{j=1}^{r}\|\mu_{0}(\cdot,\tau_{j})-\mu(\cdot,\tau_{j})\|_{2}\,.
\]
This may be interpreted as the average of the $L_{2}$ distances between the estimated and true $\tau$-conditional quantiles for $\tau\in\{\tau_{1},\ldots,\tau_{r}\}$, corresponding to a simultaneous quantile regression objective or a certain Monte Carlo approximation of $L_{2}$ distance between conditional quantile functions.

\begin{thm}\label{thm:concentration}
  Under Assumption~\ref{assn:main}, given the sequence $\varepsilon_{n}=n^{-\nu/(2\nu+d+1)}\log^{2p+1/2}n$, the posterior measure $\Pi_{n}^{(r)}$ corresponding to the generalized likelihood~\eqref{eq:iqbart-density-simultaneous} and the prior described in Assumption~\ref{assn:main}.4 with learning rate $\lambda^{-1}:=\omega_{n}\in(0, (1/2)\cdot\log^{-2p}n)$ satisfies
  \begin{equation}
    P^{n}\Pi_{n}^{(r)} \left( \left\{ \mu: \|\mu-\mu^{(r)}_{0}\|_{2,\bm{\tau}} > M_{n}\varepsilon_{n} \right\} \mid \mathcal{D}_{n} \right) \to 0
    \label{eq:concentration}
  \end{equation}
  as $n\to\infty$.
\end{thm}

\begin{proof}
  Refer to Section~\ref{sec:concentration-thm-proof}.
\end{proof}

{This result details the posterior concentration at the rate $n^{-\nu/(2\nu+d+1)}$, which matches the rate of estimation of a $\nu$-H\"{o}lder smooth function in $d+1$ dimensions under a moment condition on the noise distribution \citep{stone_82}. Notably, the flexibility attained by modeling the quantile function itself imposes a statistical cost as the dimension in the estimation rate grows to $d+1$ instead of $d$ for estimation of any particular quantile curve \citep{stone_82, syring20_gibbs_poster_concen_rates_under}.} To avoid unnecessary complexity, we prefer to analyze the single tree prior in \citep{rockova2020posterior} rather than the true BART prior of \citet{chipman2010bart}. The analysis can be extended to sums of trees by using techniques in Section 6 of \citep{rockova2020posterior} as well as to true Bayesian CART and BART priors using techniques in~\citet{rockova_saha}.

\section{  Posterior Sampling and Inference}\label{sec:sampling}

The IQ-BART posterior can be approximately sampled from using the Markov Chain Monte Carlo (MCMC) algorithm, with slight modifications to the original sampler described by~\citet{chipman2010bart}. Given datapoints $\{(\bm{X}_{i},Y_{i}):i=1,\ldots,n\}$ and a desired $r$ repetitions, we first sample $\tau_{i}^{(j)}\sim \mathcal{U}(0,1)$ independently for $i=1,\ldots,n$ and $j=1,\ldots,r$ and construct the augmented dataset $\{(\bm{X}_{i}^{(j)},Y_{i}^{(j)},\tau_{i}^{(j)}):i\in[n],j\in[r]\}$. We may subsequently employ a BART sampling algorithm of our choice, replacing the typical Gaussian likelihood on the errors with asymmetric Laplace errors with mean $0$, scale $\lambda$ and tilt $\tau_{i}^{(j)}$. For example, both the original Gibbs sampling algorithm of~\citet{chipman2010bart}, as well as the particle Gibbs algorithm in~\citet{lakshminarayanan2015particle}, may be easily modified in this way. {We provide a python implementation of an approximate posterior sampling algorithm based on particle Gibbs\footnote{\url{https://github.com/seanohagan/iqbart}}.}

We treat the scale parameter of the asymmetric Laplace likelihood, which may also be considered as the learning rate of the Gibbs posterior distribution, as a fixed hyperparameter. We found performance to be relatively robust to learning rate choices (refer to Section~\ref{sec:plug-in}), but would generally recommend that the learning rate $\lambda^{-1}$ be tuned out-of-sample or via cross-validation for optimal performance. Recent literature~\citep{syring2018calibrating} suggests a calibration strategy that may be employed to choose $b$ to yield well-calibrated frequentist coverage of credible intervals.

Our implementation leverages the particle Gibbs algorithm by \citet{lakshminarayanan2015particle} to sample trees, using standard hyperparameter choices \citep{chipman2010bart, chipman1997bayesian} to guide the tree-generating stochastic process (That is, a node at depth $d$ is nonterminal with probability $\alpha(1+d)^{-\beta}$, with $\alpha=0.95$, $\beta=2$. Splitting variables are chosen uniformly, and splits are chosen uniformly from the possible splitting points). One variation from the standard choices~\citep{chipman2010bart, chipman1997bayesian} is that we found improved performance with a larger quantity of trees. We recommend using $m=1000$ trees rather than the $m=200$ recommended in~\citep{chipman2010bart}.

A benefit of the IQ-BART framework is that it may be easily implemented in a probabilistic programming language that supports sampling trees in the style of BART~\citep{quiroga2022bart}. IQ-BART may also employ partially monotone priors on sums of trees, such as those in~\citet{chipman2022mbart}. These can foment the posterior distribution to necessarily be a distribution on proper conditional quantile functions that do not cross.

\subsection{Predictive Inference with IQ-BART}
Having access to the simulated sequence of $F_T(\cdot)$ in \eqref{eq:mu}, there are various ways in which one could extract information relevant for predicting $Y_{new}$ from  $\bm{x}_{new}\in \mathcal X$.
One possibility is to generate samples from the conditional distribution through  inverse transform sampling
\begin{equation}
Y_{new}^m= \hat Q(u_m \C \bm{X}=\bm{x}_{new}),\quad\text{where}\quad u_m\sim U[0,1],\label{eq:predict}
\end{equation}
where $\hat{Q}(\cdot\C \cdot)$ is as defined in~\eqref{eq:cond-post-mean-estimate}. This yields values $Y_{new}^m$ that are draws from an approximation to the   predictive distribution $Y\C \bm{X}=\bm{x}_{new}$. This yields a distributional forecast whose summary statistics (sample mean and sample quantiles) can yield point and interval estimators for $Y_{new}$. These point and interval estimators can be also supplied with credible intervals if one were to repeat the simulation \eqref{eq:predict} with each of the $M$ posterior draws $F_T^{(m)}(z)$ instead of the posterior mean $\hat Q(\tau \C \bm{X})$. Each of the $M$ draws emits a point and interval forecast and one can summarize the variability of these estimators for uncertainty quantification. {Such an approach may help the practitioner to decouple the aleatoric and epistemic uncertainty present in the statistic analysis: the estimated conditional distribution of $Y\mid\bm{X}$ seeks to capture uncertainty in the actual conditional distribution of the data, while the variation in the conditional model for $Y\mid \bm{X}$ across posterior draws captures uncertainty over the conditional law itself.}
Alternatively, if one is only interested in the $95\%$ credible interval for $Y_{new}$, one can use {a quantile BART procedure \cite{kindo2016bayesianquantileadditiveregression} that parameterizes the location parameter of an asymmetric Laplace likelihood corresponding to two fixed choices of $\tau\in\{2.5\%,97.5\%\}$}, yielding posterior mean estimates of the quantiles as well as credible intervals. The reported prediction interval can thus incorporate the variation in the posterior draws as well. {Our approach extends this idea by providing a means to sample from the posterior distribution over any arbitrary conditional quantile curves, yielding a Bayesian posterior over the conditional quantile function itself.}

\section{Simulation Study}\label{sec:simulation-study}

We provide empirical validation for our methodology by applying it for distributional regression in three simulated data settings. Below, we detail three examples of complex and interesting joint probability measures on covariate(s) $\bm X$ and response $Y$.

\subsection{Data-generating processes}\label{sec:data-gen-proc}

\paragraph{Difficult Conditional.}  We consider a synthetic data setting that contains difficult characteristics for distributional regression. This experiment involves varying scale, skewness, and multimodality as well as some smooth and some abrupt changes in the value of $Y$ given $X$. The model is a covariate-dependent mixture of Skew-Normal, Student $t$, and Beta random variables. The conditional distribution may be visualized in Figure~\ref{fig:jointfig}, and the data-generating process is explicitly detailed in Section~\ref{subsec:diff-cond-full}.

\paragraph{LSTAR Model.} The logistic smooth transition autoregressive (LSTAR) process has been employed for the modeling of time-dependent data in empirical finance \citep{Franses_Dijk_2000}. The process is parameterized by the vector $(\rho_{1},\rho_{2},\gamma,c,\sigma,\nu)'$ and the conditional law on a time-dependent random variable $Z_{t}$ given its previous value $Z_{t-1}$ via
\[
  Z_{t} = \rho_{1} Z_{t-1} + \rho_{2} \left(1+\exp(-\gamma(Z_{t-1}-c)\right)^{-1} Z_{t-1} + \sigma \varepsilon_{t}
\]
where $\varepsilon_{t}\overset{\mathrm{iid}}{\sim} t(\nu)$. For our experiments, we use the hyperparameter vector $(\rho_{1},\rho_{2},\gamma, c, \sigma, \nu)=(0,0.9,5,0,1,3)$. We initialize the series to begin deterministically with $Z_{0}=0$, and model the conditional distribution of $Z_{t+1}\mid Z_{t}$. Consistency with our earlier notation would therefore entail the covariate $\bm{X}$ being $Z_{t}$, and the response $Y$ being $Z_{t+1}$.

\paragraph{Covariate Dependent Mixture.} We also consider the simulation setup in \citet{orlandi2021densityregression}, which has been used to compare Bayesian nonparametric methods for distributional regression. The data comes from a univariate covariate-dependent mixture of a Gaussian and log gamma distribution, where $X_{i}\overset{\mathrm{iid}}{\sim}U(0,1)$ and $Y_{i} = f_{0}(X_{i})+\epsilon(X_{i})$. Here $f_{0}(x)=5\exp(15(x-0.5))/(1+\exp(15(x-0.5))-4x$
and $\epsilon(X_{i})$ is a mixture of a Gaussians  with mean $2x-0.6$ and variance $0.3^{3}$, and a log-Gamma distribution with shape $x^{2}+0.5$ and unit scale, with a mixing weight of $\exp(-10(x-0.8)^{2})$ on the Gaussian component. The conditional distribution changes not only in its location and spread, but also in its skewness over its domain.

\subsection{Evaluation metrics}\label{sec:eval-metrics}

Note that in each of the aforementioned simulation settings, the conditional quantiles are computable in closed form. Therefore, to evaluate the quality of estimation of a conditional distribution, we compare the estimated conditional quantiles with those of the true conditional distribution. Consider the Wasserstein-$p$ norms given by
\[
W_{p}(F,\hat{F}\mid \bm{x}) \equiv \left(\int_{0}^{1} \lvert F^{-1}(q\mid \bm{x})-\hat{F}^{-1}(q\mid \bm{x})\rvert^{p}\,\mathrm{d}q\right)^{1/p}
\]
where the limiting $p=\infty$ takes the form of the supremum. We primarily consider the average-case Wasserstein-1 error denoted by $\mathbb{E}_{\bm{X}\sim P_{X}}\left[W_{1}(F,\hat{F}\mid \bm{X})\right]$. We also consider and report additional metrics in the supplementary material, where we allow $p\in\{1,\infty\}$ and consider both the average over the covariate $\bm{X}$ as well as the supremum over $\bm{x}\in\mathcal{X}$. These metrics are computed via a Monte Carlo approximation from samples which we describe in detail below. Results from our experiments in other metrics are available in Section~\ref{sec:additional-metrics}.

\subsection{Methodology}

We proceed by sampling datasets $(\bm{X}_{1},Y_{1}),\ldots,(\bm{X}_{n},Y_{n})\overset{iid}{\sim} P_{\bm{X},Y}$ for each of the aforementioned joint probability measures $P_{\bm{X},Y}$. We obtain $B=1000$ approximate samples from the IQ-BART posterior over the conditional quantile function $\mu$, and take the posterior mean which we denote as $\hat{\mu}_{n}: \mathcal{X}\times (0,1)\to \mathbb{R}$. We approximate the previously mentioned Wasserstein-$p$ norm based error metrics by randomly sampling a grid $U_{1},\ldots,U_{G_{q}}\overset{iid}{\sim}U(0,1)$ as well as a grid $\tilde{\bm{X}}_{1},\ldots,\tilde{\bm{X}}_{G_{x}}\overset{iid}{\sim} P_{X}$. This allows us to compute a Monte-Carlo estimate of $W_{p}(F, F^{-1}\mid \tilde{\bm{X}}_{j})$ for each $\tilde{\bm{X}}_{j}$, which in turn is used for Monte-Carlo estimation of the expectation taken over $\bm{X}\sim P_{\bm{X}}$ or the supremum over $\supp P_{\bm{X}}$. For our experiments, we set $G_{q}=G_{x}=1000$.

We compare the performance of the IQ-BART posterior mean to various other methods for distributional regression, including generalized additive models for shape and scale (GAMLSS) \citep{rigby2005general}, homoscedastic and heteroscedastic BART (BART, H-BART) \citep{chipman2010bart}, density-regression BART (DRBART) \citep{orlandi2021densityregression}, implicit quantile neural networks (IQN) \citep{dabney2018implicit}, and quantile regression forests (QRF) \citep{meinshausen2006quantile}. We also compare against a random forest variant, which we deem implicit quantile forest (IQF), that uses the implicit quantile modeling strategy of \citep{dabney2018implicit} as an ablation study. When using IQ-BART, we estimate conditional quantiles using mean of the posterior predictive distribution. As demonstrated in Section~\ref{sec:plug-in}, this approach consistently and substantially outperforms the alternative plug-in estimator using the posterior mean of the conditional quantile function. We provide a comprehensive description of all experimental details for each methods in Section~\ref{sec:exp-details}.

In our experimentation, we use the same grid of quantile and covariate values for all of the methods in our comparison. We conduct the experiment for $n\in\{100,500,1000\}$, and for each value of $n$, we repeat the experiment five times to account for randomness in the training data itself, the Monte Carlo approximation of our error metrics, and in the construction of the point estimator of the conditional quantile function for each method.
\subsection{Results}

\begin{figure}[!t]
      \centering
      \includegraphics[width=0.7\linewidth]{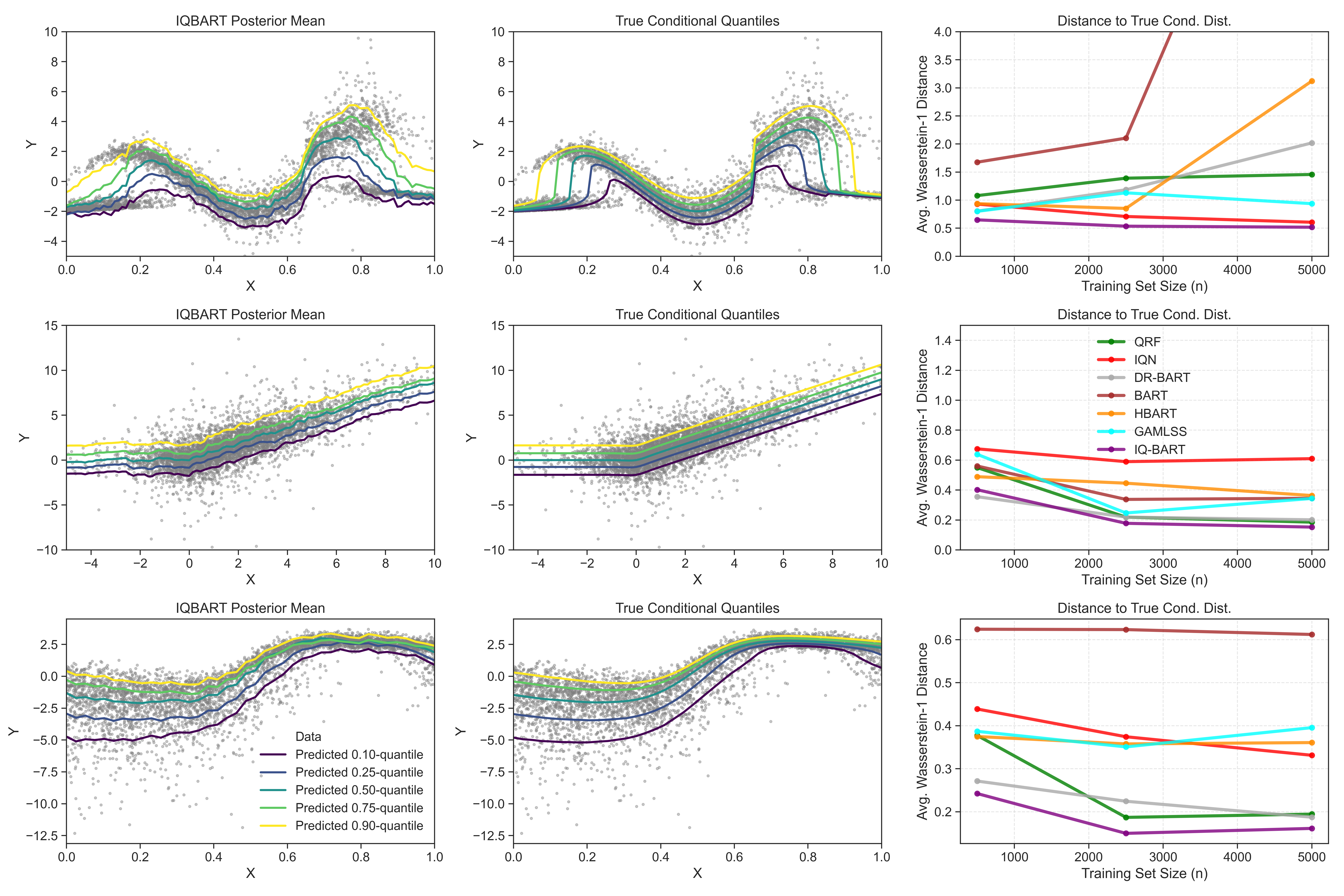}
      \caption{Results on simulated data-generating processes. The left column shows IQ-BART posterior mean estimates of the (0.1,0.25,0.5,0.75,0.9)-conditional quantiles. The center column shows true conditional quantiles. The right column compares methods for $n\in \{500,2500,5000\}$ by the average Wasserstein-$1$ distance to the true conditional distribution, displaying the mean of five replications.}\label{fig:jointfig}
\end{figure}

\begin{table}[t]
\centering
\resizebox{0.6\textwidth}{!}{%
\begin{tabular}{lccc}
\toprule
Model & Difficult Conditional & LSTAR & Cov. Dep. Mixture \\
\midrule
IQ-BART & \textbf{0.517 (0.008)} & \textbf{0.152 (0.027)} & \textbf{0.161 (0.010)} \\
QRF & 1.457 (0.088) & 0.185 (0.003) & 0.195 (0.011) \\
IQN & 0.606 (0.059) & 0.610 (0.271) & 0.331 (0.034) \\
IQF & 1.066 (0.013) & 1.601 (0.018) & 1.087 (0.022) \\
DR-BART & 2.019 (2.655) & 0.201 (0.024) & 0.187 (0.032) \\
BART & 9.813 (10.573) & 0.345 (0.043) & 0.612 (0.042) \\
H-BART & 3.122 (4.211) & 0.363 (0.011) & 0.361 (0.018) \\
GAMLSS & 0.937 (0.207) & 0.344 (0.062) & 0.396 (0.060) \\
\bottomrule
\end{tabular}}\caption{Comparison of models using the average estimated Wasserstein-1 metric across three simulated experiments. We set $n=5000$ for each task and report the mean (and 1.96 standard errors) of each metric from five repetitions. Bold values are within 1.96 SE of the best performance. Refer to Table~\ref{tab:model_comparison_full} for additional metrics.}\label{tab:model_comparison}
\end{table}
The results of our simulation study can be visualized in Figure~\ref{fig:jointfig}. For each of the three simulated data distributions, we visualize the joint distribution via a sample of size $1000$ as well as its true $\tau$-conditional quantiles for $\tau\in \{0.1,0.25,0.5,0.75,0.9\}$. We also show these estimated $\tau$-conditional quantiles via the posterior mean of IQ-BART. Finally, we show the average Wasserstein-$1$ distance to the true conditional distribution as a function of $n$, displaying the mean of three data replications. In Table~\ref{tab:model_comparison}, we display the average values of each metric across five replications, for each method in each simulation setting.  We find that IQ-BART yields salient models for the conditional distribution which result in the lowest estimated average Wasserstein-1 metric with respect to the true conditional distribution for each of these data-generating processes. Refer to Table~\ref{tab:model_comparison_full} in Section~\ref{sec:additional-metrics} for comparisons in addititional metrics.

\section{Applications to Time-Series Forecasting}

In this section, we consider data-generating processes that produce a univariate sequence $Z_{1},Z_{2},\ldots,Z_{T}$ through a $p$-th order Markov mechanism such that
\[
\pi(Z_{t}\mid Z_{t-1},\ldots,Z_{1}) = \pi(Z_{t}\mid Z_{t-1},\ldots,Z_{t-p})\,.
\]
We apply IQ-BART to model the conditional distribution of $Z_{t}\mid Z_{t-1},\ldots, Z_{t-p}$ through the posterior distribution that arises over conditional quantile function $Q(\tau\mid Z_{t-1},\ldots,Z_{t-p})$. In this section, we perform a comprehensive analysis of the predictive ability of the IQ-BART posterior through its ability to provide an interval forecast for the next datapoint in a variety of time-series, assuming a first-order Markov condition. We also employ it to nonparametrically model one-quarter-ahead financial conditions, resulting in informative posterior distributions that help to confirm previous hypotheses regarding multimodality in macroeconomic studies \citep{adrian2019multimodality}.

\subsection{Evaluation Metrics for Forecasting}\label{sec:eval-metrics}

We assess the predictive power of these predictive distributions via the \emph{continuous ranked probability score} (CRPS). For a single  observation $(\bm{X}_{i}, Y_{i})$ and predictive cumulative distribution function (CDF) $F(\cdot\mid\bm{X}_{i})$, the CRPS is defined as:
$$\text{CRPS}(\hat{F}, Y_{i}) = \int_{-\infty}^{\infty} (\hat{F}(z\mid \bm{X}_{i}) - \mathbb{I}\{z \geq y\})^2 dz\,.$$
We construct a Monte Carlo estimate of the CRPS by averaging 100 quantile scores (refer to \citet{gneiting2014probabilistic} for more information on this scoring rule). We note that the CRPS is a proper scoring rule which will be minimized in expectation by the true predictive distribution.

We also evaluate the ability of IQ-BART to construct a forecasts through the \emph{interval score}, which evaluates $1-\alpha$ predictive intervals $(L,U)$ for prediction of the subsequent true value $Y$ via
\[
\mathrm{IS}(Y, L, U) = \left((U-L)+\frac{2}{\alpha}(L-Y)\mathbb{I}\{Y<L\} + \frac{2}{\alpha}(Y-U)\mathbb{I}\{Y>U\}\right)\,.
\]
This is essentially a measure of the width of the predictive interval, with a penalty if the interval does not cover the truth. In a later application where forecasts on series of varying scales is considered, we use a mean-scaled version of the interval score given by
\[
\mathrm{MSIS}(Y_{1},\ldots, Y_{t}, Y_{t+s}, L_{t+s}, U_{t+s}) = \frac{\mathrm{IS}(Y_{t+s}, L_{t+s}, U_{t+s})}{\frac{1}{t-1}\sum_{t'=2}^{t}\lvert Y_{t'}-Y_{t'-1}\rvert}\,.
\]

\subsection{Forecasting Economic and Financial Conditions}

In this section, we apply IQ-BART to model the predictive distribution over the next quarter's real gross domestic product (GDP) quarter-on-quarter (QoQ) growth, as a function of current economic and financial conditions in the United States. We quantify financial conditions by the National Financial Conditions Index (NFCI), is a weekly index compiled by the Chicago Fed that aggregates data from money markets, debt and equity markets, and banking systems to reflect U.S. financial conditions. Economic conditions are quantified by the real gross domestic product (GDP) quarter-on-quarter (QoQ) growth.

Let $Y_{t}$ denote the real GDP QoQ at quarter $t$, and let $X_{t}$ denote the NFCI at quarter $t$. We impose the flexible model $Y_{t+1} = f(Y_{t}, X_{t}, U_{t})$ where $U_{t}\sim \mathcal{U}(0,1)$ and $f$ is a conditional quantile function, signifying a one-step-ahead dependence structure on the previous values $X_{t}$ and $Y_{t}$. We use IQ-BART, placing a BART prior on $f$ and using our augmented data approach \eqref{eq:iqbart-density-fully-augmented} to obtain a posterior distribution over $f$ as a random function given the observed data $X_{1:T}$, $Y_{1:T}$ for some $T>1$. We obtain quarterly NFCI data from the Chicago Fed and quarterly GDP data from the U.S. Bureau of Economic Analysis. The data that we consider ranges from 1972 Q4 to 2024 Q2.

\begin{figure}[t]
  \centering
  \includegraphics[width=0.65\linewidth]{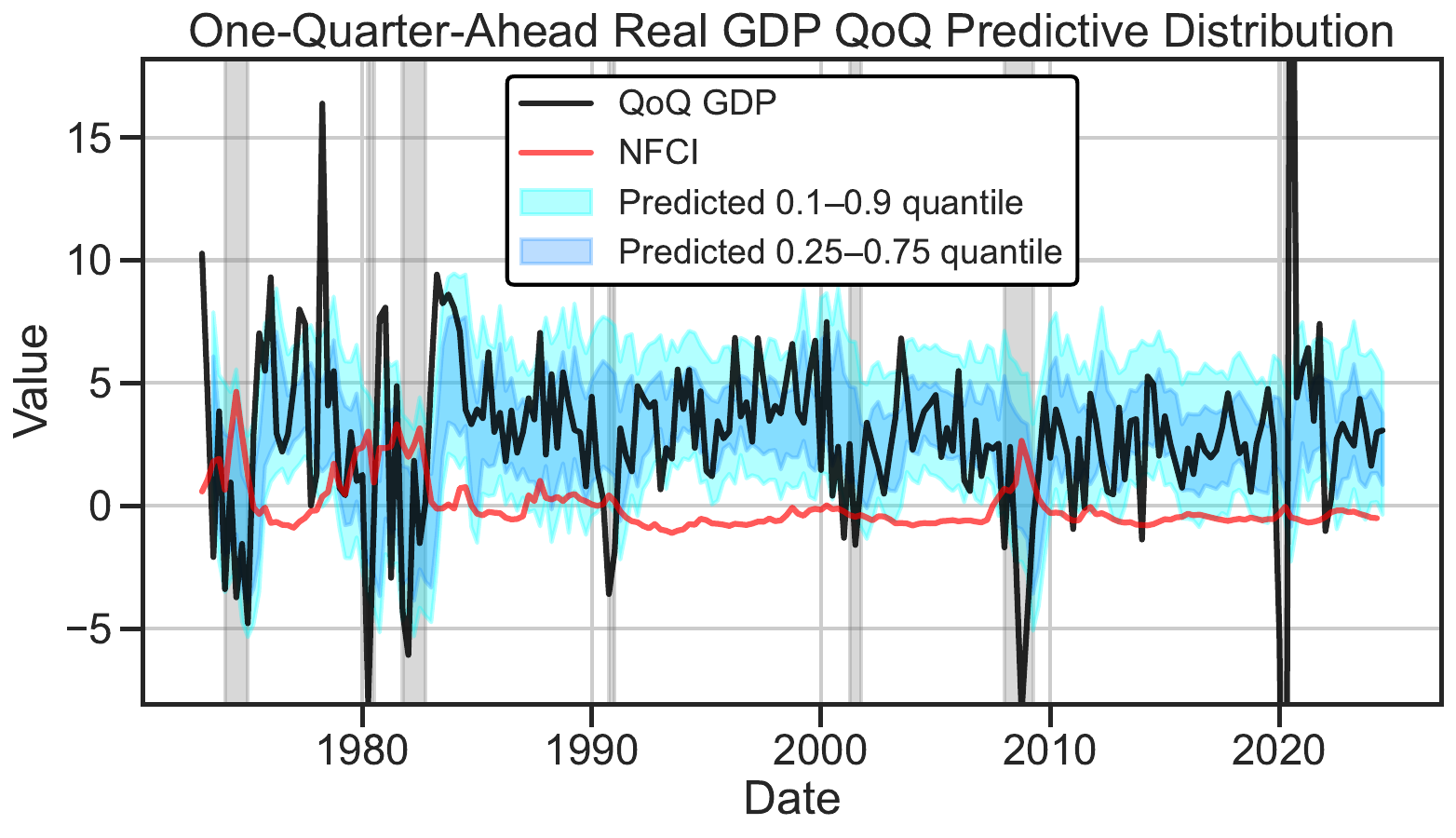}
  \caption{\label{fig:pred-rgdp-qoq} One-quarter-ahead predicted real GDP QoQ from previous quarter economic and financial conditions. Gray vertical bands denote recessions as defined by NBER.}
\end{figure}

In Figure~\ref{fig:pred-rgdp-qoq}, we observe the posterior one-quarter-ahead predictive distributions of the next point induced by IQ-BART, visualized by two credible intervals. Large values in the NFCI indicate a high uncertainty in future economic conditions. We find that the posterior distribution provides a reasonable model for the data, with the 90\% credible interval covering 93.2\% of future real GDP QoQ values, 80\% interval covering 84.4\%, and 50\% interval covering 59\%.

We additionally perform an out-of-sample evaluation of the predictive performance of IQ-BART by holding out the final fifty quarters of our observed data on which to evaluate predictive performance. Thus, our data consists of quarters indexed by $t=1,\ldots,207$ that correspond to dates from Q4 1972 to Q2 2024. For this out-of-sample evaluation, we designate quarters $t=1,\ldots,157$ as training data, and $t=158,159,\ldots,207$ as test data. This corresponds to Q1 2012 marking the start of our testing period. We sample from the approximate posterior distribution on $f$ induced by the IQ-BART model using the Gibbs sampling procedure described in Section~\ref{sec:sampling}. The approximate posterior sample of $f$ induces an approximate sample from the posterior predictive distribution $Y_{t+1}\mid Y_{t},X_{t}$ on the test data indexed by $t=158,\ldots,207$.

\begin{table}[t]
    \centering
    {\scriptsize
    \begin{tabular}{|l|c|c|c|c|c|c|c|c|c|}
        \hline
        Method & IQ-BART & AR(1) & NL-AR(1) & BART & QRF & GAMLSS & IQN & GP \\
        \hline
      CRPS & \textbf{2.95} & 3.06 & 4.04 & 3.03 & 4.44 & 2.99 & 4.55 & 3.08 \\
      IS ($\alpha=0.2$) & 38.54 & 38.82 & 73.26 & 40.81 & 44.23 & 42.57 & 45.31 & \textbf{37.04} \\
      \hline
    \end{tabular}
    }
    \caption{IQ-BART yields salient predictive distributions for one-quarter-ahead real GDP QoQ. We compare the predictive performance in terms of average CRPS and IS ($\alpha=0.2$) over the last 50 held-out datapoints (lower is better).\label{table:comp}}
  \end{table}

  In Table~\ref{table:comp} we analyze the resulting out-of-sample forecasting performance of IQ-BART when compared to various other methods in terms of the CRPS and IS (refer to Section~\ref{sec:eval-metrics}). We find that IQ-BART provides forecasts with the lowest average CRPS among the fifty held-out test points, and the second lowest interval score ($\alpha=0.2$). Our experiments found that flexible methods like IQ-BART, homoscedastic BART, and GAMLSS outperformed the linear $\mathrm{AR}(1)$ model in terms of the CRPS, giving credence to the idea that flexible non-parametric methods better capture the complexities in one-step-ahead forecasting of real GDP QoQ.

\subsection{Multimodality in One-Quarter-Ahead Financial Conditions}

Now, we consider modeling instead the next-quarter NFCI as a function of current economic and financial conditions, changing our response variable. \citet{adrian2019multimodality} find that the conditional distribution over the next-quarter NFCI exhibits complex dynamics that standard point-prediction or autoregressive models may fail to capture, particularly during periods of economic uncertainty. The authors apply a kernelized vector autoregressive model to jointly model economic and financial conditions. Recent work has also applied BART-based methods to time-series forecasting and the analysis of macroeconomic data \citep{huber2020inference, clark2023tailforecasting}. We continue in this direction by employing IQ-BART to model the predictive distribution of next-quarter NFCI. The kernel-based nonparametric model used by \citet{adrian2019multimodality} exhibits a multimodality in the predictive distribution of NFCI in particular time periods of economic uncertainty when the economy may be prone to shocks. We seek to evaluate the performance of IQ-BART in terms of both in-sample and out-of-sample evaluation, as well as to reproduce the multimodal predictive distributions that arise from the kernel-based method.

\begin{figure}[ht]
  \centering
  \begin{subfigure}[t]{0.37\linewidth}
    \centering
    \includegraphics[width=\linewidth]{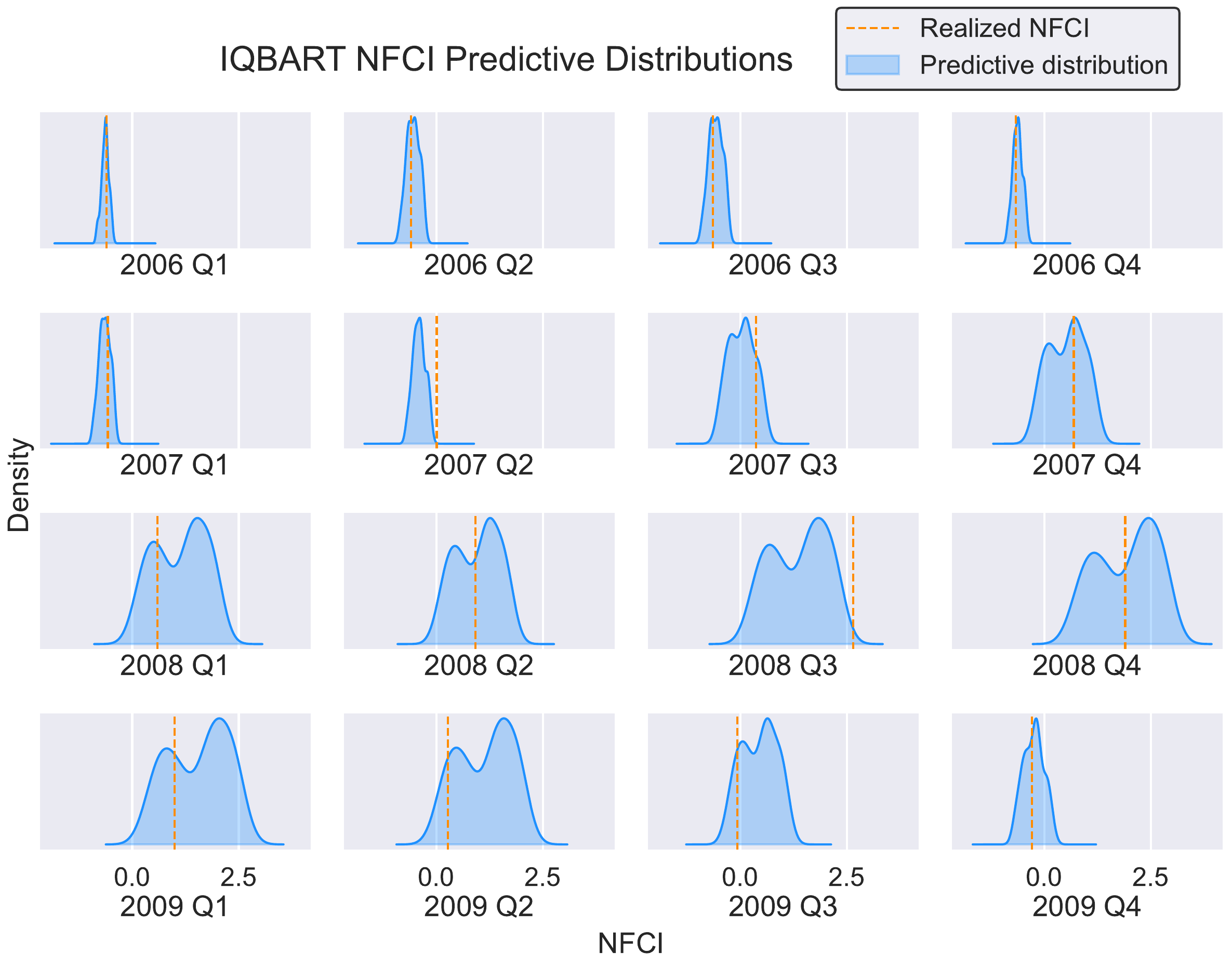}
    \caption{\scriptsize IQ-BART's next-quarter NFCI predictive densities from 2006 -- 2009. Predictive densities become multimodal in 2008, coinciding with the 2008 financial crisis.}
    \label{fig:multimodality-scatter}
  \end{subfigure}
  \hfill
  \begin{subfigure}[t]{0.45\linewidth}
    \centering
    \includegraphics[width=\linewidth]{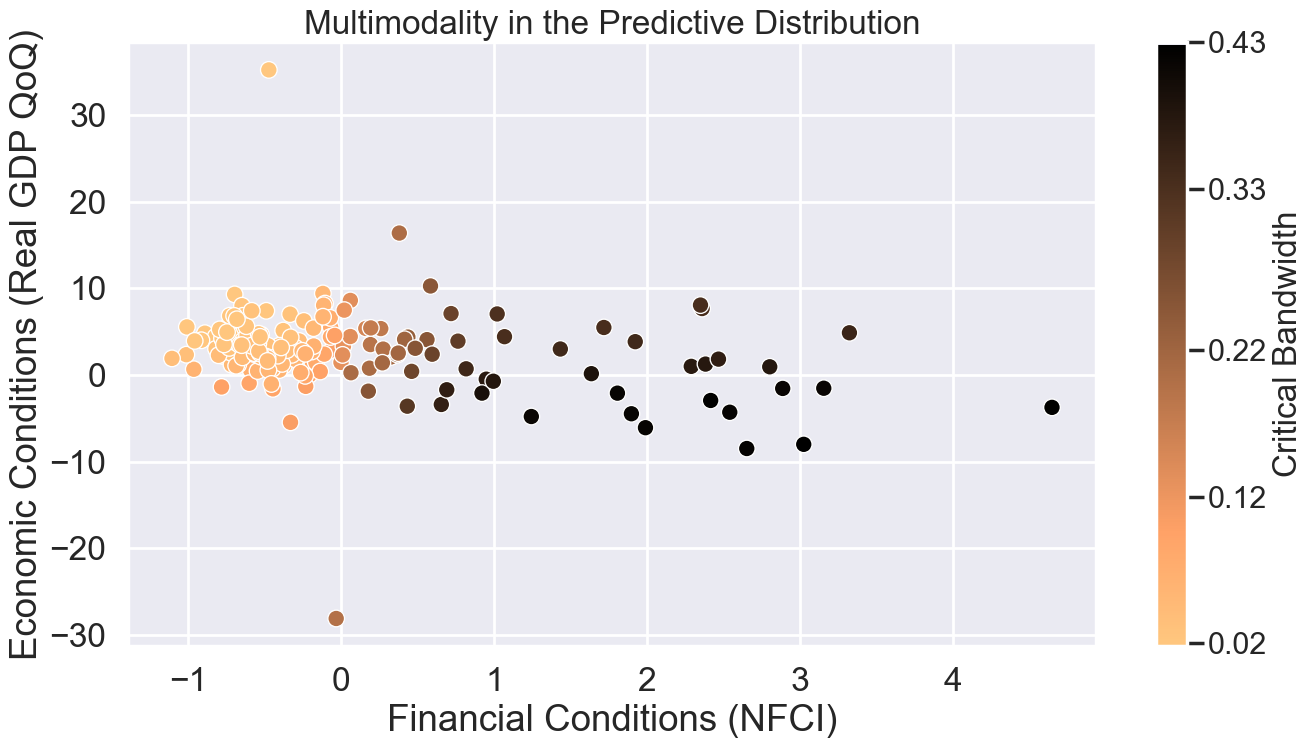}
    \caption{\scriptsize Critical bandwidth needed for a kernel density estimate of the predictive distribution to be unimodal, as a function of GDP QoQ and NFCI values}
    \label{fig:multimodality-silverman}
  \end{subfigure}
  \caption{\label{fig:in-sample-combined} Next-quarter posterior predictive distributions in NFCI arising from IQ-BART exhibit multimodal behavior in periods of poor economic and financial conditions where uncertainty is high.}
\end{figure}

Letting $t$ index quarters from 1973 Q1 to 2024 Q2, we first apply IQ-BART to learn the predictive distribution over $\mathrm{NFCI}_{t}$ as a function of $(\mathrm{NFCI}_{t-1}, \mathrm{GDP}_{t-1})$. Let $\mathcal{D}$ denote the reference table of training data consisting of tuples $(\mathrm{NFCI}_{t}, \mathrm{NFCI}_{t-1}, \mathrm{GDP}_{t-1})$ for each $t$. We sample $B=1000$ approximate samples from the posterior distribution over the conditional quantile function $Q$ using the Gibbs sampling technique describe in Section~\ref{sec:sampling}. We sample from the predictive distribution for $\mathrm{NFCI}_{t}$ by sampling from the random variable $Q(U)$, where $Q\sim \pi(\cdot\C \mathcal{D})$ is drawn from the posterior over the conditional quantile function, and $U\sim \mathcal{U}(0,1)$.

In Figure~\ref{fig:multimodality-scatter}, we visually quantify the uncertainty in the one-quarter-ahead real GDP QoQ forecast by the Silverman critical bandwidth. This represents the minimum bandwidth necessary for a KDE to be unimodal. We find that as the current value of NFCI is large, which indicates tighter financial conditions, the predictive distribution over one-quarter-ahead NFCI is more multimodal, as measured by having a large critical bandwidth. We also observe a slight gradient in terms of economic conditions, where quarters having lower real GDP growth have a tendency to have more multimodality in the next quarter's predicted NFCI as well. We find that this analysis using IQ-BART as an alternative nonparametric forecasting tool matches the conclusions found using the kernel-based method in \citep{adrian2019multimodality} without requiring the choice of a bandwidth for the nonparametric estimator.

\subsection{M4 Forecasting Competition}

\begin{table}[!t]
\centering
{\scriptsize
\begin{tblr}{
  colspec = {lccc},
  stretch = 0.1,
}
\toprule
Method & Mean MSIS & Median MSIS & Proportion Best \\
\midrule
Doornik & 8.16 & 4.53 & 0.39 \\
IQ-BART & 12.11 & 4.83 & 0.37 \\
Trotta & 9.69 & 6.99 & 0.10 \\
ETS & 8.30 & 5.78 & 0.04 \\
Roubinchtein & 12.42 & 7.99 & 0.03 \\
Naive & 9.27 & 5.68 & 0.02 \\
Smyl & 8.24 & 6.58 & 0.02 \\
Fiorucci & 8.66 & 5.74 & 0.01 \\
Ibrahim & 9.18 & 6.48 & 0.01 \\
Petropoulos & 9.19 & 5.72 & 0.01 \\
ARIMA & 9.31 & 5.82 & 0.00 \\
Montero-Manso & 9.15 & 6.89 & 0.00 \\
Segura-Heras & 13.38 & 8.84 & 0.00 \\
\bottomrule
\end{tblr}
}
\caption{IQ-BART yields competitive probabilistic forecasts when only conditioning on one past value. Comparison of MSIS (lower is better) on 100 randomly sampled daily time-series from the M4 Forecasting Competition \citep{makridakis2020m4}. }\label{tab:model_performance}
\end{table}

The M4 forecasting competition was a forecasting challenge on a collection of 100,000 time series of varying temporal spacings (hourly, daily, monthly, etc). We further analyze the predictive performance by subsampling 100 time-series from the set of daily series in the forecasting competition to judge how well the predictive performance of IQ-BART on generic time-dependent data. One objective in the competition, besides point prediction, was to give a 95\% prediction interval for the subsequent points in the time-series. In the competition, error was measured by the \emph{mean scaled interval score} (MSIS), which we define in Section~\ref{sec:eval-metrics}. The objective was to provide a method that would obtain a low average interval score across all time-series in the competition dataset. Following the methodology of~\citet{frazier2025loss}, we preprocess each time-series by differencing it until a KPSS test yields a $p$-value above 0.05, where we cannot reject the null hypothesis that the series is not stationary at a 0.05 confidence threshold. We denote this differenced series as $Y_{1},\ldots,Y_{T}$ and use IQ-BART to model the conditional probability distribution of $Y_{t}\mid Y_{t-1}$. We construct a 95\% credible interval for the single next time-step forecast for each time-series, and evaluate the quality of these intervals in terms of the MSIS. We compare these various methods which performed well in the competition and could be reproduced \citep{makridakis2020m4}. The results can be seen in Table~\ref{tab:model_performance}. We find that while IQ-BART does not perform particularly well in terms of the mean MSIS, its median MSIS is second among all methods in the competition, and enjoys having the smallest MSIS for 37\% of tested series, second also only to that of \citet{doornik2020card}, a multi-stage ensemble forecasting method. Our results are strikingly similar in this way to that of the loss-based variational predictions in~\citep{frazier2025loss}, performing well in terms of the median and proportion of time-series in which our method has the smallest MSIS. We note as well that more than only the single past point $Y_{t-1}$ may be incorporated when forecasting with IQ-BART.

\section{Conclusion}\label{sec:conclusion}
This paper develops a non-parametric quantile learning approach to distributional regression called IQ-BART. Compared to existing BART-related approaches, IQ-BART performs implicit quantile learning through data augmentation and learns the entire conditional quantile function. We investigated posterior mean point estimator under asymmetric Laplace likelihoods as a useful alternative to the routinely used posterior mode (sample) estimator and showed favorable risk properties in small samples. IQ-BART is capable of capturing multi-modal conditional distributions,  furnishes uncertainty quantification through posterior draws, extends naturally to multivariate responses, and exhibits favorable theoretical properties.

\vspace{-6mm}

\setlength{\bibsep}{2.5pt}
\renewcommand{\bibpreamble}{\linespread{0.9}\selectfont}

\bibliography{sources}

\appendix

\section{Experimental Details}\label{sec:exp-details}

\subsection{Difficult Conditional}\label{subsec:diff-cond-full}

We detail the data-generating process for the difficult conditional example. The covariate is drawn via $X\sim \mathcal{U}(0,1)$. The conditional model for $Y$ given $X$ is a covariate-dependent mixture, with both mixture weights and mixture component parameters depending on $X$. Explicitly, conditionally on $X=x$, $Y$ takes the form
\[
  Y = \begin{cases}
   Z_{1}(x) & \text{w.p. } w_{1}(x) \\
   Z_{2}(x) & \text{w.p. } w_{2}(x) \\
   Z_{3}(x) & \text{w.p. } w_{3}(x)
  \end{cases}
\]

The mixture weights $w_k(x) = \frac{\exp(l_k(x))}{\sum_{j=1}^3 \exp(l_j(x))}$ are derived from logits $l_1(x)=5-20|x-0.2|$, $l_2(x)=5-20|x-0.55|$, and $l_3(x)=5-20|x-0.85|$.

The first component is given by
\[
Z_{1}(x)\sim \mathrm{SkewNormal}(\alpha_{1}(x), \xi_{1}(x), \sigma_{1}(x))
\]
with $\alpha_{1}(x)=20(1+\exp(-15(x-0.25)))^{-1}$, $\sigma_{1}(x)=0.1+3(1+\exp(-10(x-0.15)))^{-1}$, and $\xi_{1}(x) = 1+\sin(0.8\pi\,x)-\sigma_{1}(x)\frac{\alpha_{1}(x)}{\sqrt{1+\alpha_{1}(x)^{2}}}\sqrt{\frac{2}{\pi}}$.

The second component $Z_{2}$ is a mixture of two $t$ distributions with probability 0.5 each. These are given by $t(1+\sin(0.8\pi\,x)\pm \mathrm{gap}(x), \sigma_{2}(x), 1.2)$ where $\mathrm{gap}(x)=6(1+\exp(-15(x-0.4)))^{-1}$ and $\sigma_{2}(x)=0.1+0.2(1+\exp(-10(x-0.3)))^{-1}$.

Finally, the third component $Z_{3}$ is a mixture of a Beta and a t distribution given by
\[
Z_{3}(x) = 0.85\left(\sigma_{B}(x)(B+\lambda_{B}(x))\right) + 0.15\,t(\mu_{t}(x), \sigma_{t}(x), 1.5)
\]
where $B\sim\mathrm{Beta}(a_{B}(x), b_{B}(x))$. The relevant parameter functions are given by $a_B(x) = \max(0.05, 0.2+0.3\sin(2.5\pi(x-0.6)))$, $b_B(x) = \max(0.05, 0.2+0.3\cos(2.5\pi(x-0.6)))$, $\sigma_B(x) = \max(0.5, 2+2\sin(2\pi(x-0.6)))$, and $\lambda_{B}(x)$ set such that the mean of the mixture component is $1+\sin(0.8\pi\,x)-3(1+\exp(-10(x-0.7) ))^{-1}$.
The Student $t$ sub-component has location $\mu_t(x) = 1+\sin(0.8\pi\,x)-3(1+\exp(-10(x-0.7) ))^{-1}+3+2\sin(3\pi(x-0.6))$, and scale $\sigma_t(x) = 0.3+0.2(1+\exp(-15(x-0.8)))^{-1}$.

\subsection{Simulation Study}
We detail the experimental hyperparameters for each method used in the simulation study.

\paragraph{Quantile Regression Forests} For our experiments, we use a random forest with 200 trees and a maximum depth of 5 for each tree. The random forest implementation employs bootstrap sampling of the training data for each tree, meaning each tree is trained on a random subset of data points sampled with replacement. Additionally, each tree uses a random subset of features, specifically 80\% of the available features are randomly selected for each tree (though in simulated data experiments, we have just one feature). The algorithm utilizes mean aggregation to combine predictions from individual trees. Each decision tree in the forest optimizes splits based on variance reduction, with a minimum decrease threshold of 0.01 controlling the tree growth (splits must reduce variance by at least this amount to be considered valid). For each tree, the conditional distribution at a point $\bm{x}$ is estimated by the empirical distribution of the training responses falling into the same leaf node as $\bm{x}$. Quantiles are then calculated from this empirical distribution and averaged across trees in the forst (see~\citep{meinshausen2006quantile}).

\paragraph{Implicit Quantile Neural Network} We use an Implicit Quantile Network (IQN) \citep{dabney2018implicit} with a feed-forward neural network architecture consisting of 3 hidden layers with [128, 128, 128] units respectively, each followed by leaky ReLU activation functions. The network is trained to minimize the implicit quantile regression loss after augmenting the data with uniformly sampled quantile values for each datapoint.
The training process uses the Adam optimizer with a learning rate of 1e-3, running for 20 epochs with a batch size of 64. We sample one quantile value per data point during training.

\paragraph{Implicit Quantile Random Forest} We also employ a standard random forest counterpart of our method, which performs the same data augmentation scheme as in the implicit quantile network to pair each datapoint with a uniformly sampled quantile value. Then, the tree-growing criterion is based on the reduction of the implicit quantile regression loss. This method is included mainly for the purposes of ablation.

\paragraph{Implicit Quantile BART} For IQ-BART, we use 1000 trees with $\alpha=0.95$ and $\beta=2.0$. We sample the BART components using a particle Gibbs sampler~\citep{lakshminarayanan2015particle}, as implemented for probabilistic programming using \texttt{pymc}~\citep{quiroga2022bart}. We use 10 quantile repetitions per datum using the data-augmentation scheme corresponding to~\eqref{eq:iqbart-density-fully-augmented} and use a fixed learning rate of 1. We note that tuning the learning rate can potentially further improve results (see Section~\ref{sec:plug-in}).

\paragraph{Density Regression BART} We use the default implementation provided by \citet{orlandi2021densityregression}. We draw 1000 burn-in samples and 1000 posterior samples for density regression BART, using a grid size of 1000 for the response variable.


\paragraph{BART} We use 200 trees, with $\alpha=0.95$ and $\beta=2.0$, using the particle Gibbs sampler~\citep{lakshminarayanan2015particle}. We use homoscedastic Gaussian errors with a half-normal prior with scale 5. We sample 1000 burn-in samples and sample 1000 draws from the posterior. Conditional quantiles are infered via the Gaussian quantile function.

\paragraph{Heteroscedastic BART} For heteroscedastic BART, we model both the mean and the log-variance of a Gaussian distribution as functions of the input features. We use 50 trees for each BART model (one for the mean and one for the log-variance), with tree prior parameters $\alpha=0.95$ and $\beta=2.0$. We sample 1000 burn-in samples and draw 1000 posterior samples. Trees are sampled using the particle Gibbs sampler~\citep{lakshminarayanan2015particle}. Conditional quantiles are infered by first computing the posterior mean over the mean and variance and then using the Gaussian quantile function.

\paragraph{GAMLSS} We use a normal distribution family with penalized B-splines to model both the mean and standard deviation as functions of the input features. The model is trained with a maximum of 50 iterations to ensure convergence. We use the R GAMLSS package. Conditional quantiles are infered via the Gaussian quantile function.

\subsection{Predicting Economic Conditions}
We compare against a linear autoregressive model using only the previous values for real GDP QoQ and NFCI, denoted by AR(1). We also compare against a kernel-based nonlinear autoregressive model using the specifications in \citep{adrian2019multimodality}, denoted as NL-AR(1). G-BART designates a BART model for the conditional mean as a sum of trees with homoscedastic Gaussian noise. QRF designates quantile regression forests \citep{meinshausen2006quantile}. GP denotes a Gaussian process with squared exponential kernel. GAMLSS denotes a Gaussian additive model for shape and scale. Finally, IQN denotes an implicit quantile neural network using a four layer feedforward architecture.

\section{Additional Experiments}

\subsection{Additional Metrics in Simulation Study}\label{sec:additional-metrics}

Recall the Wasserstein-$p$ norms given by
\[
W_{p}(F,\hat{F}\mid x) \equiv \left(\int_{0}^{1} \lvert F^{-1}(q\mid x)-\hat{F}^{-1}(q\mid x)\rvert^{p}\,\mathrm{d}q\right)^{1/p}\,,
\]
where the limiting $p=\infty$ takes the form of the supremum. In the main text, we primarily consider the average-case taken over both the covariate and quantile. In this supplementary section, we consider $p\in\{1,\infty\}$, and consider the average-case error $\mathbb{E}_{X\sim P_{X}}\left[W_{p}(F,\hat{F}\mid X)\right]$ as well as the worst-case error $\sup_{x\in\mathcal{X}}W_{p}(F,\hat{F}\mid x)$. Choosing $p\in\{1,\infty\}$ and either the average or supremum allows us to capture four distinct error metrics: the average ($p=1$) and worst-case ($p=\infty$) errors with respect to quantile levels, each evaluated either in expectation over the feature space or at the worst-case feature point. This provides a comprehensive assessment of estimation quality across different dimensions of robustness.

\begin{table}[!ht]
\centering
\resizebox{\textwidth}{!}{%
\begin{tabular}{lcccccccccccc}
\toprule
& \multicolumn{4}{c}{Difficult} & \multicolumn{4}{c}{LSTAR} & \multicolumn{4}{c}{Mixture} \\
\cmidrule(lr){2-5} \cmidrule(lr){6-9} \cmidrule(lr){10-13}
Model & $\sup W_\infty$ & $\sup W_1$ & $\mathbb{E}[W_\infty]$ & $\mathbb{E}[W_1]$ & $\sup W_\infty$ & $\sup W_1$ & $\mathbb{E}[W_\infty]$ & $\mathbb{E}[W_1]$ & $\sup W_\infty$ & $\sup W_1$ & $\mathbb{E}[W_\infty]$ & $\mathbb{E}[W_1]$ \\
\midrule
IQ-BART & \textbf{4.30 (0.145)} & \textbf{1.36 (0.023)} & \textbf{1.85 (0.214)} & \textbf{0.52 (0.008)} & 1.11 (0.236) & \textbf{0.39 (0.113)} & 0.68 (0.159) & \textbf{0.15 (0.027)} & 2.96 (0.301) & \textbf{0.54 (0.042)} & 1.41 (0.223) & \textbf{0.16 (0.010)} \\
QRF & 5.97 (0.064) & 2.97 (0.107) & 3.88 (0.207) & 1.46 (0.088) & 3.61 (1.150) & 1.00 (0.250) & 0.91 (0.148) & 0.19 (0.003) & 3.38 (0.480) & 0.92 (0.174) & 1.18 (0.165) & 0.19 (0.011) \\
IQN & 5.79 (1.669) & 1.55 (0.270) & 2.10 (0.092) & 0.61 (0.059) & 3.40 (0.319) & 1.26 (0.670) & 2.47 (0.347) & 0.61 (0.271) & 3.27 (1.072) & 0.75 (0.136) & 1.26 (0.152) & 0.33 (0.034) \\
IQF & 6.76 (0.013) & 2.80 (0.033) & 2.97 (0.074) & 1.07 (0.013) & 7.80 (0.066) & 4.33 (0.049) & 2.91 (0.076) & 1.60 (0.018) & 3.31 (0.524) & 1.55 (0.101) & 1.96 (0.042) & 1.09 (0.022) \\
DR-BART & 8.92 (5.231) & 3.53 (2.854) & 5.51 (6.065) & 2.02 (2.655) & \textbf{1.02 (0.323)} & 0.64 (0.290) & \textbf{0.44 (0.090)} & 0.20 (0.024) & \textbf{2.59 (0.309)} & 0.64 (0.214) & \textbf{0.84 (0.033)} & 0.19 (0.032) \\
BART & 30.99 (31.089) & 11.05 (10.938) & 28.53 (30.553) & 9.81 (10.573) & \textbf{0.93 (0.152)} & \textbf{0.42 (0.076)} & 0.69 (0.074) & 0.34 (0.043) & 4.28 (0.396) & 0.96 (0.152) & 2.79 (0.078) & 0.61 (0.042) \\
H-BART & 65.09 (99.937) & 22.10 (34.954) & 9.52 (12.729) & 3.12 (4.211) & 2.40 (1.213) & 1.07 (0.511) & 0.94 (0.141) & 0.36 (0.011) & 3.16 (0.337) & 0.74 (0.220) & 1.80 (0.154) & 0.36 (0.018) \\
GAMLSS & 8.29 (2.751) & 2.85 (0.271) & 2.68 (0.857) & 0.94 (0.207) & 1.14 (0.488) & \textbf{0.46 (0.155)} & 0.73 (0.184) & 0.34 (0.062) & 3.16 (0.417) & 0.69 (0.156) & 1.82 (0.101) & 0.40 (0.060) \\
\bottomrule
\end{tabular}
}\caption{Comparison of models using the Wasserstein metrics across three simulated experiments. We set $n=5000$ for each task and report the mean (and 1.96 standard errors) of each metric from five repetitions. Bold values are within 1.96 SE of the best performance.}\label{tab:model_comparison_full}
\end{table}

Table~\ref{tab:model_comparison_full} displays a more detailed comparison than Table~\ref{tab:model_comparison} on the three simulated data-generating processes in Section~\ref{sec:simulation-study}. This table includes all of the aforementioned four metrics for a more detailed comparison.

\subsection{Learning Rate Comparison}

\begin{figure}[h!t]
  \centering
\includegraphics[width=\linewidth]{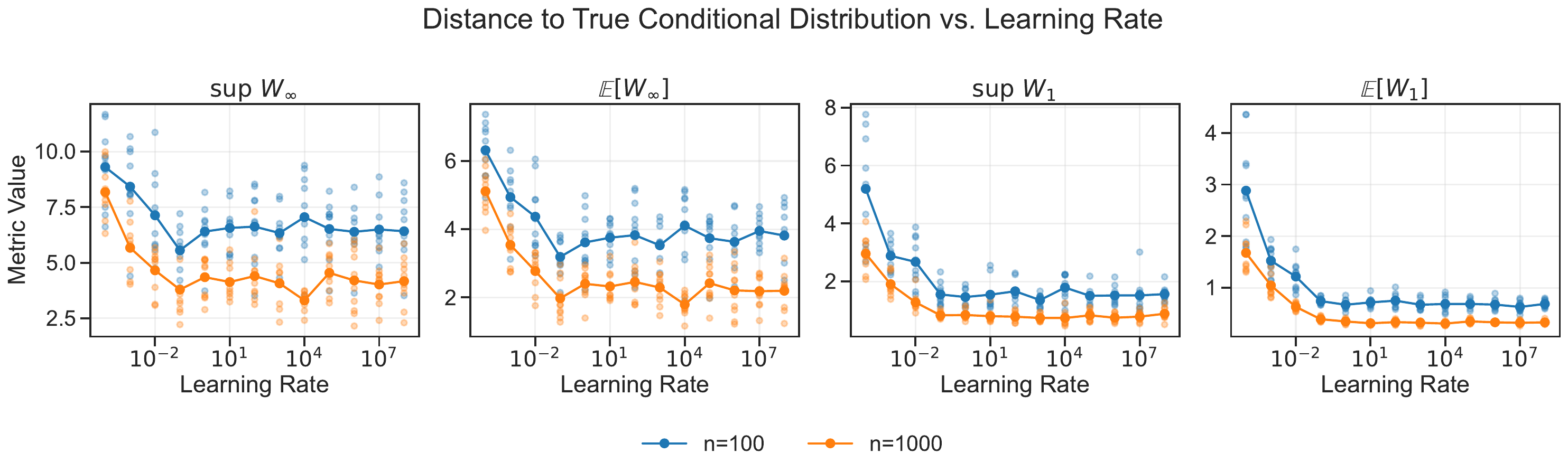}
  \caption{\label{fig:lr-comp}Comparison of modeling quality in terms of the four metrics described in Section~\ref{sec:eval-metrics} for varying values of the learning rate $\lambda$.}
\end{figure}

Figure~\ref{fig:lr-comp} displays the quality of modeling the conditional distribution by the IQ-BART posterior mean as a function of the learning rate $\lambda$. We evaluate the quality in terms of the four metrics described in Section~\ref{sec:eval-metrics}. These experiments are performed using $n=1000$ samples from the covariate-dependent mixture joint distribution (see Section~\ref{sec:data-gen-proc}). For each learning rate value we perform ten repeated experiments, and we display the mean value of each metric. Our general conclusion is that provided the learning rate is sufficiently large (i.e. at least 0.1), we do not notice a strong effect in terms of the modeling quality.

\subsection{Quantile Repetition Comparison}

\begin{figure}[h!t]
  \centering
\includegraphics[width=\linewidth]{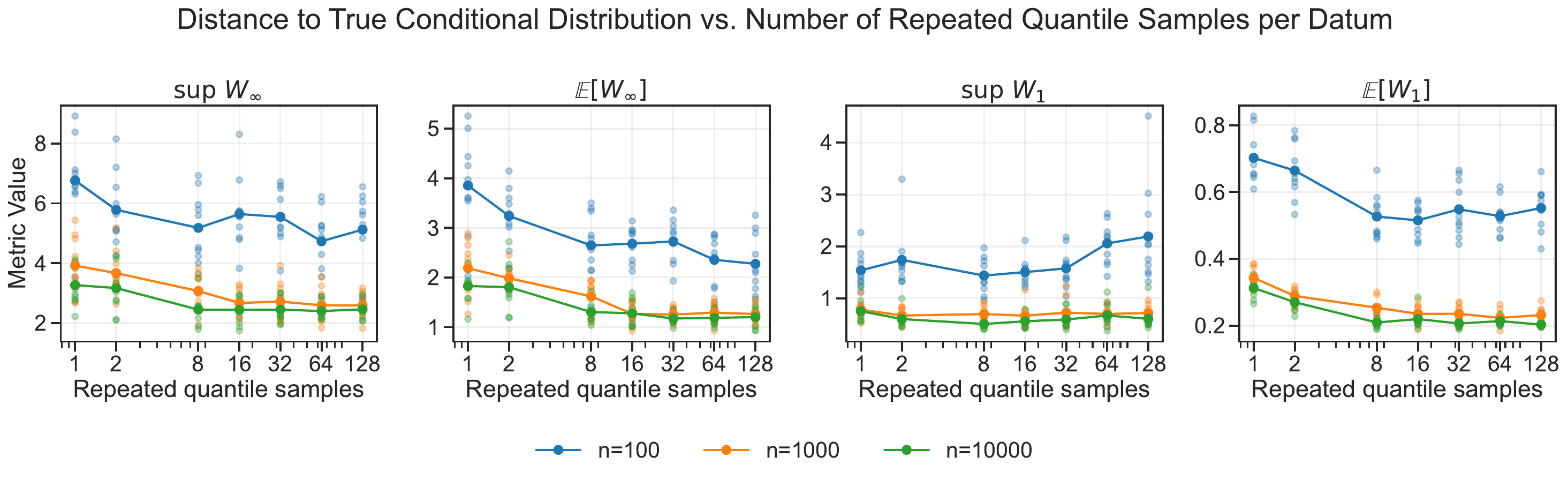}
  \caption{\label{fig:n-rep-comp}Comparison of modeling quality in terms of the four metrics described in Section~\ref{sec:eval-metrics} for varying values of the number of quantile repetitions per datum.}
\end{figure}

Figure~\ref{fig:n-rep-comp} displays the quality of modeling the conditional distribution by the IQ-BART posterior mean as a function of the number of quantile repetitions per data point. We evaluate the quality in terms of the four metrics described in Section~\ref{sec:eval-metrics}. These experiments are performed using $n\in\{10^{2},10^{3},10^{4}\}$ samples from the covariate-dependent mixture joint distribution (see Section~\ref{sec:data-gen-proc}). For each combination of $n$ and a number of repetitions $r$, we perform ten repeated experiments. We display the mean value of each metric as a separate line corresponding to each $n$. We generally conclude that increasing the number of quantile repetitions improves performance.

\subsection{Posterior predictive vs plug-in}\label{sec:plug-in}

\begin{table}[t]
\centering
\resizebox{\textwidth}{!}{%
\begin{tabular}{llccc}
\toprule
Learning Rate & Method & Difficult & LSTAR & Mixture \\
\midrule
0.1 & Posterior mean plug-in & 0.611 (0.040) & 0.231 (0.019) & 0.239 (0.017) \\
 & Posterior predictive mean & 0.519 (0.007) & 0.137 (0.009) & 0.165 (0.014) \\
  \midrule
1 & Posterior mean plug-in & 0.593 (0.036) & 0.199 (0.022) & 0.207 (0.017) \\
 & Posterior predictive mean & 0.517 (0.008) & 0.152 (0.027) & 0.161 (0.010) \\
\midrule
10 & Posterior mean plug-in & 0.589 (0.018) & 0.201 (0.023) & 0.208 (0.013) \\
 & Posterior predictive mean & 0.519 (0.009) & 0.151 (0.015) & 0.150 (0.010) \\
\midrule
100 & Posterior mean plug-in & 0.582 (0.024) & 0.225 (0.026) & 0.211 (0.020) \\
 & Posterior predictive mean & 0.526 (0.012) & 0.160 (0.028) & 0.151 (0.010) \\
\bottomrule
\end{tabular}
}
\caption{Performance comparison of posterior predictive mean vs plug-in estimation from the posterior mean for IQ-BART. We display the mean (1.96 standard errors) of the $\mathbb{E}[W_{1}]$ metric (refer to Section~\ref{sec:eval-metrics}) from five trials using IQ-BART on our three example data-generating processes in Section~\ref{sec:data-gen-proc}. We also compare for different learning rate values $\lambda^{-1}\in\{0.1,1,10,100\}$.}\label{tab:iqbart_comprehensive}
\end{table}

In Table~\ref{tab:iqbart_comprehensive}, we investigate the performance discrepancy between using the posterior predictive mean to estimate conditional quantile values versus using the plug-in estimator from the posterior mean over the conditional quantile function. We also repeat this comparison for different learning rate values, and for each of our example data-generating processes. We conclude that estimation via the mean of the posterior predictive distribution over the quantile uniformly outperforms plug-in estimation by the posterior mean, and to a very substantial degree.

\section{Theory}

\subsection{Notation and Definitions}

We let $\mathcal{V}^{K}$ denote the set of tree-structured step functions on $\mathcal{X}\times (0,1)\to \mathbb{R}$ with $K$ terminal nodes. We use the notation $\lesssim$ and $\gtrsim$ to denote inequality up to a universal constant. Given a function $g:A\times B\to C$ and any fixed $b\in B$, we use $g(\cdot,b)$ to denote the partial application from $A$ to $C$ given by $a\mapsto g(a,b)$.

\begin{defn}

  We define the \emph{single quantile loss} $\ell(\bm{X},Y;\tau)$ as the usual check-loss in estimation of the single conditional quantile $\tau$ by the $\tau$-conditional quantile candidate $h:\mathcal{X}\to\mathbb{R}$ given by
  \[
    \ell(\bm{X},Y,h;\tau)=\rho_{\tau}(h(\bm{X})-Y)\,.
  \]
  We also define the \emph{expected quantile loss} $\ell(\bm{X},Y,\mu)$ as the average check-loss incurred by a uniformly sampled quantile value for the single datapoint $\bm{X},Y$. This is expressed as
  \begin{equation}
    \ell(\bm{X}, Y, \mu) := \int_{0}^{1} \rho_{\tau}(\mu(\bm{X},\tau) - Y)\,\mathrm{d}\tau.
    \label{eq:single-expected-quantile-loss}
  \end{equation}
  Similarly, the \emph{expected quantile risk} $R(\mu)$ of a conditional quantile function estimator $\mu: \mathcal{X}\times (0,1)\to \mathbb{R}$ is given by the expectation
  \begin{equation}
    \ell(\mu) := \mathbb{E}_{\bm{X},Y\sim P_{\bm{X},\bm{Y}}}\left[\int_{0}^{1} \rho_{\tau}(\mu(\bm{X},\tau) - Y)\,\mathrm{d}\tau\right].
    \label{eq:empirical-expected-quantile-risk}
  \end{equation}
  Finaly, we also define its empirical version given $n$ datapoints $\bm{X},\bm{Y}$ as
  \begin{equation}
    \ell_{n}(\mu;\bm{X}, \bm{Y}) := \frac{1}{n}\sum_{i=1}^{n}\int_{0}^{1} \rho_{\tau}(\mu(X_{i},\tau) - Y_{i})\,\mathrm{d}\tau.
    \label{eq:empirical-expected-quantile-risk}
  \end{equation}

  We define the $\bm{r}$\emph{-Monte Carlo quantile loss} as a Monte Carlo approximation of the above by $r$ samples
  \begin{equation}
    \ell^{(r)}(\bm{X}, Y, \mu; \bm{\tau}) := \frac{1}{r}\sum_{j=1}^{r} \rho_{\tau_{j}}(\mu(\bm{X},\tau_{j}) - Y)
    \label{eq:single-mc-emp-quantile-loss}
  \end{equation}
  Similarly, we define the $\bm{r}$\emph{-Monte Carlo quantile risk} via the expectation
  \begin{equation}
    \ell^{(r)}(\mu; \bm{\tau}) := \frac{1}{r}\sum_{j=1}^{r}\mathbb{E}_{\bm{X},\bm{Y}}\left[ \rho_{\tau_{j}}(\mu(\bm{X},\tau_{j}) - Y)\right]
    \label{eq:mc-emp-quantile-risk}
  \end{equation}
  and, its empirical version given $n$ data points as
  \begin{equation}
    \ell_{n}^{(r)}(\mu;\bm{X}, \bm{Y}, \bm{\tau}) := \frac{1}{nr}\sum_{i=1}^{n}\sum_{j=1}^{r} \rho_{\tau_{j}}(\mu(X_{i},\tau_{j}) - Y_{i})
    \label{eq:mc-emp-quantile-risk}
  \end{equation}
  where $\tau_{1},\ldots,\tau_{r}\overset{\mathrm{iid}}{\sim} \mathcal{U}(0,1)$.
\end{defn}

\begin{lemma}\label{lem:approximation}

  Under Assumption~\ref{assn:main}, there exists a tree-structured step function $\mu_{\widehat{\mathcal{T}},\widehat{\beta}}\in\mathcal{V}^{K}$ for some $K=2^{s\,(d+1)}$ with $s\in\mathbb{N}\setminus\{0\}$ such that
  \[
\|\mu_{0}-\mu_{\widehat{\mathcal{T}},\widehat{\beta}}\|_{\infty}\leq \|\mu_{0}\|_{\mathcal{H}^{\alpha}} C M^{\alpha} (d+1)/K^{\alpha/(d+1)}
  \]
\end{lemma}

\begin{proof}
  Following the proof of Lemma 3.2 in~\citet{rockova2020posterior}, for any valid tree partition $\mathcal{T}=\{\Omega_{k}:k\in[K]\}$, there exists a step function $\mu_{\mathcal{T}, \widehat{\beta}}$ such that for any $k\in [K]$, for any $\bm{x}\in \Omega_{k}\cap \left(\mathcal{X}\times [0,1]\right)$, we have
  \[
    \lvert \mu(\bm{x}) - \mu_{\mathcal{T}, \widehat{\beta}}(\bm{x})\rvert \leq \|\mu\|_{\mathcal{H}^{\alpha}}\mathrm{diam}^{\alpha}(\Omega_{k}; \mathcal{S})\,.
  \]
  Taking $\widehat{\mathcal{T}}$ to be a $k$-$d$ tree partition $\{\widehat{\Omega}_{k}:k\in [K]\}$ (refer to \citep{rockova2020posterior}), and taking the supremum over $\bm{x}\in \mathcal{X}$ yields
  \[
    \|\mu-\mu_{\widehat{\mathcal{T}}, \widehat{\beta}}\|_{\infty} < C_{\mathrm{max}}\|\mu\|_{\mathcal{H}^{\alpha}} \max_{1\leq k\leq K}\mathrm{diam}^{\alpha}(\widehat{\Omega}_{k}; \mathcal{S})
  \]
  for some absolute constant $C_{\mathrm{max}}>1$. The same technique used in the proof of Theorem 3.2 of~\citet{rockova2020posterior} then yields the desired conclusion.
\end{proof}

\subsection{Proof of Theorem~\ref{thm:concentration}}\label{sec:concentration-thm-proof}

\begin{proof}

  We employ Theorem 4.1 of \citet{syring20_gibbs_poster_concen_rates_under}. The authors provide a set of sufficient conditions on our prior distribution and loss function to achieve concentration of the Gibbs posterior around $\mu_{0}$ at a particular rate. We verify these sufficient conditions given the setting in Assumption~\ref{assn:main}.

  \paragraph{Condition 2.} Observe that the loss function $\ell^{(r)}$ (defined in~\eqref{eq:single-mc-emp-quantile-loss}) is of the form
  \[
    \ell^{(r)} = \sum_{j=1}^{r}\ell(\cdot\,;\tau_{j})\,.
  \]
  From the proof of Proposition 8 in~\citep{syring20_gibbs_poster_concen_rates_under}, there exists $K>0$ such that for $\mu \in \mathcal{F}_{n}$, the loss $\ell(\cdot\,;\tau): \mathcal{F} \to \mathbb{R}$ satisfies
  \[
    \mathbb{E}_{\bm{X},Y\sim P}\left[\exp\left(-\omega_{n} \left( \ell(\bm{X},Y,\mu;\tau) - \ell(\bm{X},Y,\mu_{0}(\cdot,\tau);\tau) \right)\right) \right] \leq \exp \left( -K_{n} \omega_{n} \varepsilon_{n}^{2} \right)
  \]
  with $K_{n}=(K/2)\log(n)^{-2p}$. We note that this requires Assumption~\ref{assn:main}.3. Therefore, we have for $\mathcal{F}_{n}=\{\mu:\|\mu\|_{\infty}\leq \log^{p}(n)\}$,
  \begin{align*}
   &\mathbb{E}_{\bm{X},Y\sim P}\left[\exp\left(-\omega_{n} \left( \ell^{(r)}(\bm{X},Y,\mu) - \ell^{(r)}(\bm{X},Y,\mu_{0}) \right)\right)\right] \\ =\,\,& \mathbb{E}_{\bm{X},Y\sim P}\left[\exp\left(-\frac{1}{r}\sum_{j=1}^{r} \omega_{n} \left( \ell(\bm{X},Y,\mu(\cdot,\tau_{j});\tau_{j}) - \ell(\bm{X},Y,\mu_{0}(\cdot,\tau_{j});\tau_{j}) \right)\right)\right]
    \\\leq \,\,& \frac{1}{r}\sum_{j=1}^{r} \mathbb{E}_{\bm{X},Y\sim P}\left[ \exp\left(- \omega_{n} \left( \ell(\bm{X},Y,\mu(\cdot,\tau_{j});\tau_{j}) - \ell(\bm{X},Y,\mu_{0}(\cdot,\tau_{j});\tau_{j}) \right)\right)\right] \\
  \leq \,\,& \exp \left( -K_{n} \omega_{n} \varepsilon_{n}^{2} \right)
  \end{align*}
  which verifies Condition 2 of~\citep{syring20_gibbs_poster_concen_rates_under} with $\omega_{n}=c\log^{-2p}(n)$ for some $0<c<1/2$.

  \paragraph{Prior mass condition.} Now, we verify the prior mass condition in Theorem 4.1 of~\citep{syring20_gibbs_poster_concen_rates_under}.
  As in~\citep{syring20_gibbs_poster_concen_rates_under}, we define $m(\mu,\mu_{0}):=\mathbb{E}_{\bm{X},Y\sim P}\left[\ell^{(r)}(\bm{X},Y,\mu)-\ell^{(r)}(\bm{X},Y,\mu_{0})\right]$ and $v(\mu,\mu_{0}):= \mathbb{E}_{\bm{X},Y\sim P} \left[ \left(\ell^{(r)}(\bm{X},Y,\mu) - \ell^{(r)}(\bm{X},Y,\mu_{0})\right)^{2} \right]-m(\mu,\mu_{0})^{2}$ to denote the mean and variance of the excess loss. Similarly, define the mean and variance of the excess $\tau$-check loss specifically as $m_{\tau}(\mu,\mu_{0}):=\mathbb{E}_{\bm{X},Y\sim P}\left[\ell(\bm{X},Y,\mu(\cdot,\tau);\tau)-\ell(\bm{X},Y,\mu_{0}(\cdot,\tau);\tau)\right]$ and $v_{\tau}(\mu,\mu_{0}):= \mathbb{E}_{\bm{X},Y\sim P} \left[ \left(\ell(\bm{X},Y,\mu(\cdot,\tau);\tau) - \ell(\bm{X},Y,\mu_{0}(\cdot,\tau);\tau)\right)^{2} \right]-m_{\tau}(\mu,\mu_{0})^{2}$.

  We note that
  \begin{align*}
    m(\mu,\mu_{0}) &= \frac{1}{r}\sum_{j=1}^{r}m_{\tau_{j}}(\mu,\mu_{0}) \\ &\lesssim \frac{1}{r}\sum_{j=1}^{r}\left(\frac{1}{\tau_{j}(1-\tau_{j})}\right)\|\mu(\cdot,\tau_{j})-\mu_{0}(\cdot, \tau_{j})\|_{L_{1}} \\
    &\lesssim \left(\frac{1}{r}\sum_{j=1}^{r}\frac{1}{\tau_{j}(1-\tau_{j})}\right) \max_{j\in[r]}\|\mu(\cdot,\tau_{j})-\mu_{0}(\cdot, \tau_{j})\|_{L_{1}}\,.
  \end{align*}
  and similarly
  \begin{align*}
    v(\mu,\mu_{0}) &\leq \frac{1}{r}\sum_{j=1}^{r}\left(v_{\tau_{j}}(\mu,\mu_{0}) + m_{\tau_{j}}(\mu,\mu_{0})^{2}\right) \\ &\lesssim \frac{1}{r}\sum_{j=1}^{r}\left(\left(\frac{1}{\tau_{j}^{2}(1-\tau_{j})^{2}}\right)\left(\|\mu(\cdot,\tau_{j})-\mu_{0}(\cdot, \tau_{j})\|^{2}_{L_{2}} + \|\mu(\cdot,\tau_{j})-\mu_{0}(\cdot,\tau_{j})\|_{L_{1}}^{2}\right)\right)\\
    &\lesssim \left(\frac{1}{r}\sum_{j=1}^{r}\frac{1}{\tau_{j}^{2}(1-\tau_{j})^{2}}\right) \left(\max_{j\in[r]}\|\mu(\cdot,\tau_{j})-\mu_{0}(\cdot, \tau_{j})\|^{2}_{L_{2}} + \max_{j\in[r]}\|\mu(\cdot,\tau_{j})-\mu_{0}(\cdot, \tau_{j})\|^{2}_{L_{1}}\right)\,.
  \end{align*}

  We now proceed as in~\citet{rockova2020posterior}. Lemma~\ref{lem:approximation} guarantees that there exists a tree-supported step-function $\mu_{\widehat{\mathcal{T}},\widehat{\beta}}$ on a partition $\widehat{\mathcal{T}}$ that approximates $\mu_{0}$ with error
  \[
    \|\mu_{0}-\mu\|_{\infty} \leq \|\mu_{0}\|_{\mathcal{H}^{\nu}} C_{1} (d+1) / K^{-\nu/(d+1)}\,.
  \]
  We choose $K=J_{n}$ such that the approximation error is bounded above by $\epsilon_{n}/2$. The reverse triangle ineqality and this choice of $J_{n}$ then yield the conclusion that $\|\beta-\widehat{\beta}\|_{\infty}<\epsilon_{n}/2$ implies that $\|\mu_{0}-\mu_{\widehat{\mathcal{T}},\beta}\|_{\infty}<\epsilon_{n}$.

  We can proceed using the argument in the proof of Proposition 8 of~\citet{syring20_gibbs_poster_concen_rates_under}. Define $J_{n}=n^{-1/(2\nu+d+1)}$. We now consider tree-supported step functions $\mu_{\wh{\mathcal{T}},\beta}$ on a partition $\wh{\mathcal{T}}$ and a coefficient vector $\beta\in \mathbb{R}^{J_{n}}$.
  The above bounds on $m(\mu,\mu_{0})$ and $v(\mu,\mu_{0})$ allow us to conclude that for any $\mu_{\wh{\mathcal{T}},\beta}$ satisfying $\|\mu_{\wh{\mathcal{T}},\beta}-\mu_{0}\|_{\infty}\leq CJ_{n}^{-\nu}$, we have
  \[
    m(\mu_{\wh{\mathcal{T}},\beta},\mu_{0})\vee v(\mu_{\wh{\mathcal{T}},\beta},\mu_{0}) \lesssim J_{n}^{-2\nu}\,.
  \]
  This implies that
  \begin{align*}
    \Pi\left\{\mu\in \mathcal{F}_{\mathcal{T}} : m(\mu,\mu_{0})\vee v(\mu,\mu_{0}) \leq J_{n}^{-2\nu}\right\} &\geq  \Pi\left\{\mu\in \mathcal{F}({\wh{\mathcal{T}}}) : m(\mu,\mu_{0})\vee v(\mu,\mu_{0}) \lesssim J_{n}^{-2\nu}\right\}\pi(\widehat{\mathcal{T}}) \\                                                  &\geq  \Pi\left\{\mu\in \mathcal{F}({\wh{\mathcal{T}}}) : \|\mu_{0}-\mu\|_{\infty} \lesssim J_{n}^{-\nu}\right\}\pi(\widehat{\mathcal{T}}) \\
     &\geq  \Pi\left\{\widehat{\beta}\in\mathbb{R}^{J_{n}} : \|\beta_{0}-\widehat{\beta}\|_{\infty} \lesssim J_{n}^{-\nu}\right\}\pi(\widehat{\mathcal{T}}) \\
    &\gtrsim \exp(-C_{2}J_{n}\log n)
  \end{align*}
  for some constant $C_{2}>0$, where the last inequality bounds the (log of the) first factor exactly as in the proof of Theorem 1 of~\citet{ shen2015adaptive}. The (log of the) second factor is bounded as in the proof of Theorem 4.1 in~\citet{rockova2020posterior}. This verifies that the step-function prior on the conditional quantile function $\mu$ assigns sufficient mass around the truth.

  Thus, the two sufficient conditions for Theorem 4.1 of~\citet{syring20_gibbs_poster_concen_rates_under} hold, yielding the desired result. Note that the condition on the sequence of prior supports is true by Assumption~\ref{assn:main}, where we restrict the support of the prior to $\mathcal{F}_{n}$, but that this may not be necessary (see \citet{kleijn2006misspecification}).

\end{proof}

\subsection{Proof of Lemma~\ref{lem:flat-pm}}\label{sec:pf-lem-1}

\begin{proof}

Recall that $q^{\mathrm{PM}}_{\tau,\lambda}$ takes the form~\eqref{eq:quantile-post-mean}. The numerator can be evaluated as
\begin{align*}
   &\sum_{k=0}^{n} \int_{Y_{(k)}}^{Y_{(k+1)}} \mu \exp\left(-\lambda^{-1}(\tau-1) \sum_{i=1}^{k}Y_{(i)}-\lambda^{-1}\tau\sum_{i=k+1}^{n}Y_{(i)}\right)\exp\left(-\lambda^{-1}(k-n\tau)\mu\right)\,\mathrm{d}\mu \\ =&
    \sum_{k=0}^{n} \int_{Y_{(k)}}^{Y_{(k+1)}} w_{k}\exp\left(-c_{k}\mu\right)\mu\,\mathrm{d}\mu \\
=& \sum_{k=0}^{n} w_{k} \begin{cases}\frac{1}{c_{k}^{2}}\left[(c_{k}Y_{(k)}+1)\exp(-c_{k}Y_{(k)})-(c_{k}Y_{(k+1)}+1)\exp(-c_{k}Y_{(k+1)})\right] & k/n \neq \tau \\ \frac{1}{2}\left[Y_{(k+1)}^{2}-Y_{(k)}^{2}\right] & k/n = \tau\end{cases}
\end{align*}
where $w_{k}:=\exp\left(-\lambda^{-1}(\tau-1) \sum_{i=1}^{k}Y_{(i)}-\lambda^{-1}\tau\sum_{i=k+1}^{n}Y_{(i)}\right)$ and $c_{k}:= \lambda^{-1}(k-n\tau)$. We also have $Y_{(0)}:=-\infty$ and $Y_{(n+1)}:=\infty$.

The denominator can be evaluated in a similar way
\begin{align*}
   &\sum_{k=0}^{n} \int_{Y_{(k)}}^{Y_{(k+1)}} \exp\left(-\lambda^{-1}(\tau-1) \sum_{i=1}^{k}Y_{(i)}-\lambda^{-1}\tau\sum_{i=k+1}^{n}Y_{(i)}\right)\exp\left(-\lambda^{-1}(k-n\tau)\mu\right)\,\mathrm{d}\mu \\ =&
    \sum_{k=0}^{n} \int_{Y_{(k)}}^{Y_{(k+1)}} w_{k}\exp\left(-c_{k}\mu\right)\,\mathrm{d}\mu \\
=& \sum_{k=0}^{n} w_{k} \begin{cases}\frac{1}{c_{k}}\left[\exp(-c_{k}Y_{(k)})-\exp(-c_{k}(Y_{(k+1)})\right] & k/n \neq \tau \\ \left[Y_{(k+1)}-Y_{(k)}\right] & k/n = \tau\end{cases}
\end{align*}

The form~\eqref{eq:pm-form} allows for numerically stable computation of the estimator in log-space and allows the practitioner to avoid approximate inference via numerical integration or MCMC methods.

\end{proof}

\subsection{Proof of Lemma~\ref{lem:posterior_mean_convergence}}\label{subsec:asymp-pf}

\begin{proof}
We refer to \citet{chernozhukov2003classical} and \citet{ruppert2015statistics} which characterize the asymptotic behavior of Gibbs posteriors and of the sample quantile respectively.
\end{proof}

\subsection{Proof of Lemma~\ref{lem:asymp-lam}}\label{subsec:asymp-pf-lam}

\begin{proof}
  As discussed in Sec. 3.1 of \citet{bissiri2016general}, the limiting behavior as the learning rate grows large and as the learning rate converges to zero for proper priors is well known.

In the case of the improper prior, as $\lambda \to \infty$, the exponent approaches zero, and thus $\lim_{\lambda \to \infty} w_k = 1$ for all $k=0, \ldots, n$. We use the change of variables $\nu=\mu-\bar{Y}$ to write \begin{align}\label{eq:mean-and-remainder}
  \hat{q}^{\mathrm{PM}}_{\tau,\lambda}&=\frac{\int (\nu + \bar{Y})\exp\left(-\lambda^{-1}\sum_{i=1}^{n}(Y_{i}-\bar{Y}-\nu)(\tau-\mathbb{I}\{Y_{i}-\bar{Y}<\nu\})\right)\mathrm{d}\nu}{\int \exp\left(-\lambda^{-1}\sum_{i=1}^{n}(Y_{i}-\bar{Y}-\nu)(\tau-\mathbb{I}\{Y_{i}-\bar{Y}<\nu\})\right)\mathrm{d}\nu} \notag \\
  &=\bar{Y}+ \frac{\int \nu\exp\left(-\lambda^{-1}\sum_{i=1}^{n}(Z_{i}-\nu)(\tau-\mathbb{I}\{Z_{i}<\nu\})\right)\mathrm{d}\nu}{\int \exp\left(-\lambda^{-1}\sum_{i=1}^{n}(Z_{i}-\nu)(\tau-\mathbb{I}\{Z_{i}<\nu\})\right)\mathrm{d}\nu}
\end{align}
where $Z_{i}=Y_{i}-\bar{Y}$ for $i=1,\ldots,n$ are the centered data points. We turn out attention to the second term and expand the numberator as
\begin{align*}
\sum_{k=1}^{n-1} w_{k}\int_{Z_{(k)}}^{Z_{(k+1)}}\nu \exp\left(-c_{k}\nu\right)\,\mathrm{d}\nu &+ w_{0}\int_{-\infty}^{-M}\nu \exp\left(-c_{0}\nu\right)\,\mathrm{d}\nu \\ &+ w_{0}\int_{-M}^{-Z_{1}}\nu \exp\left(-c_{0}\nu\right)\,\mathrm{d}\nu \\ &+ w_{n}\int_{Z_{(n)}}^{M}\nu \exp\left(-c_{n}\nu\right)\,\mathrm{d}\nu \\&+ w_{0}\int_{M}^{\infty}\nu \exp\left(-c_{n}\nu\right)\,\mathrm{d}\nu
\end{align*}
where $M>0$ is a constant such that $M>Z_{(n)}$ and $-M < Z_{(1)}$. The integrands on all interior terms between $-M$ and $M$ converge uniformly on finite intervals, and so the dominated convergence theorem allows the limits to be evaluated inside the integrals, causing the sum of all of these terms to converge to a finite constant. We recall that $w_{k}\to 1$ as $\lambda\to\infty$ for all $k=0,1,\ldots,n$, and so the asymptotic behavior is controlled by the outer terms. These evaluate to
\[
\int_{M}^{\infty}\nu\exp(-c_{n}\nu)\,\mathrm{d}\nu=\frac{\lambda}{n(1-\tau)}M\exp\left(-\frac{n(1-\tau)\,M}{\lambda}\right)+\frac{\lambda^{2}}{n^{2}(1-\tau)^{2}\exp\left(-\frac{n(1-\tau)\,M}{\lambda}\right)}
\]
and
\[
\int_{-\infty}^{-M}\nu\exp(-c_{0}\nu)\,\mathrm{d}\nu=-\frac{\lambda}{n\tau}M\exp\left(-\frac{n\tau M}{\lambda}\right)-\frac{\lambda^{2}}{n^{2}\tau^{2}}\exp\left(-\frac{n\tau M}{\lambda}\right)
\]
As $\lambda\to\infty$, the numerator therefore becomes asymptotically becomes of order $\lambda^{2}$ for $\tau\neq 0.5$, and for $\tau=0.5$ becomes of constant order with respect to $\lambda$. Similar analysis of the denominator shows that it is always of order $\lambda$ for all $\tau\in(0,1)$. This implies that the remainder term in~\eqref{eq:mean-and-remainder} is of order $\lambda$ for $\tau\neq 0.5$ and of order $\lambda^{-1}$ for $\tau=0.5$.

\end{proof}

\end{document}